\documentclass[11pt]{article}

\usepackage[utf8]{inputenc}
\usepackage[T1]{fontenc}
\usepackage{microtype}

\usepackage[english]{babel}

\usepackage{charter}
\usepackage{bbm}
\usepackage{csquotes,xpatch} %

\usepackage[numbib,nottoc]{tocbibind} 
\usepackage{amsfonts}
\usepackage{amsmath}
\usepackage{amssymb}
\usepackage{amsthm}
\allowdisplaybreaks[2]
\usepackage{mathtools}

\usepackage{macro}

\usepackage{hyperref}
\hypersetup{colorlinks=true,
            breaklinks=true,
            linkcolor=[rgb]{0, 0.2, 0.5},
            citecolor=[rgb]{0, 0.35, 0},      
            urlcolor=[rgb]{0, 0.2, 0.5},
            pdfpagemode=UseOutlines}
\usepackage[nameinlink]{cleveref}
\usepackage{authblk}  
\theoremstyle{plain}
\newtheorem{theorem}{Theorem}
\newtheorem*{theorem*}{Theorem}
\newtheorem{proposition}[theorem]{Proposition}
\newtheorem*{proposition*}{Proposition}
\newtheorem{corollary}[theorem]{Corollary}
\newtheorem*{corollary*}{Corollary}
\newtheorem{lemma}[theorem]{Lemma}
\newtheorem*{lemma*}{Lemma}
\newtheorem{conjecture}{Conjecture}
\newtheorem*{conjecture*}{Conjecture}

\newtheorem*{result*}{Result}

\newtheorem*{question*}{Question}
\theoremstyle{definition}
\newtheorem{definition}{Definition}
\newtheorem*{definition*}{Definition}

\newtheorem*{example*}{Example}

\newtheorem*{problem*}{Problem}
\theoremstyle{remark}
\newcommand{\cH}{\mathcal{H}}
\newcommand{\enc}{\mathsf{Enc}}
\newcommand{\dec}{\mathsf{Dec}}

\newcommand{\preim}{\mathsf{PreIm}}

\newcommand{\poly}{\mathrm{poly}}

\newtheorem*{remark*}{Remark}

\usepackage{ifthen}  
\newif\ifExtAbs
\ExtAbsfalse
\ifExtAbs
\else
\fi
\ifExtAbs
  \usepackage[inline,shortlabels]{enumitem}
  \usepackage[margin=0.75in]{geometry}
\else
  \usepackage{enumitem}
  \usepackage[margin=1in]{geometry}
\fi

\title{ \vspace{-2cm}Towards Universal Quantum Tamper Detection}
\ifExtAbs
  \author{Anne Broadbent}
\author{Upendra Kapshikar}
\author{Denis Rochette }

\affil{}
\else
\author{Anne Broadbent \thanks{anne.broadbent@uottawa.ca}}
\author{
Upendra Kapshikar \thanks{ukapshik@uottawa.ca}}
\author{Denis Rochette \thanks{denis.rochette@uottawa.ca}}

\affil{Department of Mathematics and Statistics, University of Ottawa}
\fi
\date{}

\date{\vspace{-1cm}}

\begin{document}

\maketitle


\ifExtAbs
\else

  \begin{abstract}

Tamper-resilient cryptography studies how to protect  data against adversaries who can physically manipulate codewords before they are decoded.
The notion of tamper detection codes formalizes this goal, 
{requiring that any 
unauthorized modification}
be detected with high probability. 
Classical results, starting from Jafargholi and Wichs (TCC 2015), established the existence of such codes against very large families of tampering functions—subject to structural restrictions ruling out identity and constant maps. Recent works of Boddu and Kapshikar (Quantum, 7) and Bergamaschi (Eurocrypt 2024)  have extended these ideas to quantum adversaries, but  
 only consider unitary tampering families.

In this work, we give the first general treatment of quantum tamper detection against arbitrary {quantum maps.} 
We show that Haar-random encoding schemes achieve exponentially small soundness error against any adversarial family whose size, Kraus rank, and entanglement fidelity obey natural constraints, which are direct quantum analogues of the min-entropy and fixed-point restrictions in the classical setting. Our results unify and extend previous work, subsuming both the classical and unitary-only adversarial families.

Beyond this, we demonstrate a fundamental separation between classical and quantum tamper detection. Classically, relaxed tamper detection (which allows either rejection or recovery of the original message) cannot protect even against the family of constant functions.
This family is of size $2^n$.
In contrast, we show that quantum encodings can handle this obstruction, and we conjecture and provide evidence that they may in fact provide relaxed tamper detection and non-malleable security against any family of  quantum
maps of size up to $2^{2^{\alpha n}}$ for any constant $\alpha <\frac{1}{2}$,  {leading to our conjecture on the existence of what we call \emph{universal} quantum tamper detection}.  {Taken together, our results provide the first  evidence that quantum tamper detection is strictly more powerful than its classical counterpart.} 
\end{abstract}
 \newpage
\tableofcontents
 \newpage
\fi

\ifExtAbs
\else
  \section{Introduction}
\fi
Cryptographic security is typically studied under the assumption that the adversary sees only the inputs and outputs of a cryptographic device—not its internal state. In many realistic scenarios, however, an attacker can physically tamper with the stored data (flip bits, inject phases, apply local gates, etc.) and then observe the device’s outputs. The field of \emph{tamper-resilient security} asks how to encode data so that such tampering is either detected or in some cases, rendered harmless.

 \emph{Tamper detection codes \cite{JW15}} formalize this goal via tampering experiments with respect to a family of adversarial functions $\advCPTP$:
\ifExtAbs
  \begin{enumerate}[(1)]
\else
  \begin{enumerate}
\fi
    \item An encoder maps a message $m \in \mathcal{M}$ into a codeword $c \in \mathcal{C}$.
    \item An adversary applies a tampering function $f$ from some family of functions $\mathcal{F}_{\Adv}$ (from $\mathcal{C}$ to $\mathcal{C}$) to produce $\hat{c}=f(c)$.
    \item Given $f(c)$, a decoder  produces $\widehat{m} \in \mathcal{M}\ \cup \lbrace \perp \rbrace$.
\ifExtAbs
  \end{enumerate}
\else
  \end{enumerate}
\fi

A natural and minimal correctness requirement is that when there is no tampering, the decoder always recovers $m$ (\emph{completeness}).
The security requirement is that for every nontrivial tampering, the decoder outputs~$\bot$ with high probability over the scheme’s internal randomness (\emph{soundness}).
\ifExtAbs
\else
  Formally, these conditions can be expressed as follows:
\begin{description}[leftmargin=!,labelwidth=\widthof{\bfseries Completeness}]
\item [Completeness]\label{item:intro_completeness_1} $\Pr\left[\dec\left(\enc(m)\right)=m\right]=1$.	
\item[Soundness] \label{item:intro_soundness_1} $\Pr\left[\dec \left( f \left( \enc(m) \right) \right)= \perp \right] \geq 1-\varepsilon$.

\end{description} 	
\fi

A notable class tamper detection codes are Algebraic Manipulation Detection (AMD) codes \cite{CDF+08,CPX15}: they provide optimal/near-optimal protection against the additive tampering families $\mathcal{F}_{\Adv} = \{f_e \text{ such that }\ f_e : x \mapsto x \oplus e\}$.

For tamper detection codes, one can ask how large and complex can we allow $\mathcal{F}_{\Adv}$ to be, while maintaining tamper security.
After all, the total number of functions from $\lbrace 0,1 \rbrace^n$ is $2^{2^n}$ but the additive adversarial family (in the case of AMD codes) is only of size $2^n$.
One can see that it is not possible to provide tamper security against the full family of $2^{2^n}$ functions, simply because there exists \emph{bad} families for tamper detection. 
For example, one can not detect tampering against the identity function.
Suppose such a scheme existed and we try to encode message $m$ using it.
Then, the completeness condition requires the decoder to output $\widehat{m}=m$ whereas the security condition requires the decoder to output $\widehat{m}=\perp$. Hence, such a scheme can not exist.
Of course, this example is rather artificial as one does not expect to catch an adversary who does nothing (applies the identity).
The other set of bad families is that of constant functions. 
Let $f$ be a constant function that maps every input to $\enc(m)$ for some fixed message $m$.
Now consider two cases, first where we intended to encode the message $m$ and the adversary does nothing.
Second, where we intended to encode the message $m^\prime$ and adversary tampered this with $f$.
In both cases, the decoder receives $\widehat{c}=\enc{(m)}$, but in the first case we need the decoder to output $\widehat{m}=m$ (due to completeness); whereas in the second case, we need $\widehat{m}=\perp$ (due to security). 
Thus, such a scheme that protects against constant functions can not exist either.

Jafargholi and Wichs~\cite{JW15} showed that (as far as classical adversaries are concerned), these are the only \emph{bad} families. 
In particular, using probabilistic methods they proved the existence of tamper detection codes against very large families (even up to almost doubly-exponential size), provided each tampering function satisfies two structural conditions — it has high output min-entropy when the input is uniform, and it has few fixed points.
Intutively, the first condition says that each function should be far from being a constant function and the second says that it should be far from the identity function.
One can note that unlike AMD codes, these proofs for tamper detection against large families use probabilistic methods, and efficient constructions are not known in general. 

Boddu and Kapshikar \cite{BK23} extended tamper detection to the case of an adversary who is capable of performing quantum operations. 
Firstly, they construct a quantum analogue of AMD codes, where $\mathcal{F}_{\Adv}$ is the family of generalized   Pauli operators. 
They showed that if $\mathcal{F}_{\Adv}$ is the family of unitary operators such that $\vert\mathcal{F}_{\mathsf{Adv}} \vert < 2^{2^{\alpha n}}$  for some constant $\alpha \in (0,1/6)$ and  each $U \in \mathcal{F}_{\Adv}$ is far from the identity map (under some norm).
Similar to classical constructions, quantum analogues for AMD codes are efficient whereas for large doubly exponential sized family, the proof remains existential.
However, a general quantum operation is not necessarily unitary.
{Notably, since the results of \cite{BK23} are limited to only unitary operators, the classical results in the literature and quantum versions presented by \cite{BK23} remain largely incomparable; in fact, the only common operators in the intersection of classical functions and unitary operators are permutations.  
}
Thus, the state of the art is that tamper detection against general quantum adversaries remains essentially wide open. 

\ifExtAbs
\else
 \subsection{Our contributions}
\fi

Our main contributions in the work can be broadly separated into two directions.  
\begin{itemize}
    \item Firstly, we extend tamper detection beyond prior classical and unitary-only settings to fully general quantum operations (CPTP maps).
Under natural structural constraints on the adversarial family, we prove that Haar-random encoding schemes (\cref{fig:Haar Encoding}) achieve tamper detection with exponentially small error.
The conditions that we use can be considered to be \CPTP analogues of the conditions for classical maps. 
We thus achieve a unified framework which incorporates both classical as well as quantum results known in the literature (\cite{JW15} and \cite{BK23}, respectively). 
In addition, we show that by slightly tweaking the decoder,  a similar result for tamper detection holds for \emph{quantum} messages, with a minor loss in parameters.

\item Secondly, we propose and give evidence for the conjecture that quantum tamper detection can achieve security in ways that any classical encoding can not.
In particular, we conjecture that quantum encodings may achieve relaxed\footnote{As the name suggests, relaxed tamper detection is a relaxation of the soundness condition of the tamper detection, where the decoder can either abort or output the original untampered message\ifExtAbs
\else
(see Definition~\ref{def:relaxed_TD} for a formal definition.)
\fi
} 
tamper detection against superexponential sized families without any constraints ---  leading to what we call  \emph{universal quantum tamper detection}. We contrast this with classical encodings in the relaxed setting, which are known to be impossible against any exponential-sized family.  {In support of the conjecture, we show a quantum advantage in encoding a family of size $2^n$, which is known to be impossible classically.}   
\end{itemize}

\ifExtAbs
\else
Before expanding on  the above contributions, we first explain the setup and model of tamper detection with quantum encodings.

\subsubsection*{Model and security goal}

Let $\mathcal{M}$ be the message space and let $\mathcal{H} \cong \CC^d$.  
A (public, keyless) tamper detection scheme against an adversarial family $\advCPTP$ is a pair $(\enc,\dec)$ with
\[
\enc: \mathcal{M} \to \mathcal{H}, \qquad \dec: \mathcal{B}(\mathcal{H}) \to \mathcal{M} \cup \{\bot\},
\]
such that:
\begin{align*}
     \Pr \left[ \dec(\enc(m))=m \right] &= 1 & \forall m \in \mathcal{M} && \text{(completeness)} \\
     \Pr \left[ \dec \circ \Phi \circ \enc(m) = \perp \right] &\geq 1 - \varepsilon & \forall \Phi \in \advCPTP,\, \forall m \in \mathcal{M} && \text{(soundness)}
\end{align*}

Our goal is thus to identify broad, natural conditions on a family $\advCPTP$ under which tamper detection is possible, specifically by showing $\enc$ and $\dec$ procedures that provide completeness and exponentially small soundness parameter,~$\epsilon$. 

With the setup in place, we can now move onto our results, divided into two main parts.  

\fi
\subsubsection*{Contribution Part A: Generalization of tamper detection to \CPTP maps 
}

\begin{theorem}[Informal]
\label{thm:informal-main}
Let $\advCPTP$ be an adversarial family of CPTP maps on $\mathcal{B}(\CC^d)$. Suppose that for constants $0 \leq 2\alpha < \delta < 1$ the family satisfies:
\begin{enumerate}
    \item $|\advCPTP| \leq 2^{d^\alpha}$;
    \item every $\Phi \in \advCPTP$ admits a Kraus representation with Kraus rank at most $d^{1-\delta}$;
    \item the maximal entanglement fidelity across the family is bounded by 
    \[
    F_e(\advCPTP) \leq 2^{-d^{\delta/2}}.
    \]
    Then the Haar-random encoding (\cref{fig:Haar Encoding}) is a tamper detection code against $\advCPTP$.
\end{enumerate}
\end{theorem}
The first condition in the above theorem says that the adversarial family should not be too large.
As previously discussed, it is well known that one can not construct tamper detection codes against arbitrarily large families.
 {Intuitively, the second condition says that $\Phi$ is not close to any constant map. 
 The third condition captures that $\Phi$ is not too close to the identity map.}

\ifExtAbs
\else

This statement is formalized in Theorem~\ref{thm:tamper_detection_code}, which implies that the Haar-random  encoding $\enc$ into $\CC^d$, together with the natural projector-measurement decoder, is $\varepsilon$-tamper secure against $\advCPTP$ with $\varepsilon=2^{-\Omega(k)}$. We thus have that under simple, interpretable constraints on adversarial CPTP maps, quantum tamper detection codes exist with exponentially small soundness error.

As previously mentionned, one of the main motivations of extending tamper detection to channels is to give a unified analysis of classical tamper detection results as well as unitary tamper detection results.
In that regard, the following corollaries can be obtained from Theorem~\ref{thm:tamper_detection_code}:

\noindent
\begin{corollary}[Classical adversary form, informal]
    Fix integers $n,k$ and write $d=2^n$. Let $\mathcal F$ be any family of classical maps 
$f:\{0,1\}^n\to\{0,1\}^n$. Assume every tampering map $f\in\mathcal F$ satisfies two conditions: 

\begin{enumerate}
    \item[(i)] \emph{High min-entropy:} the distribution $f(X)$ has entropy 
    \[
      H_\infty(f(X)) \ge n\delta
    \]
    \item[(ii)] \emph{Few fixed points:} 
    \[
      \Pr[f(X)=X] \;\le\; \sqrt{2}\,2^{-n\delta/4}.
    \]
\end{enumerate}
\end{corollary}

If the family $\mathcal F$ is not too large --- specifically, if $|\mathcal F|\le 2^{2^{\alpha n}}$ for some appropriately chosen $\alpha$, then a Haar-random choice of an $d\times d$ unitary $U$ yields an encoding/decoding pair $(\enc_U,\dec_U)$ that is tamper-secure against every $f\in\mathcal F$.

The above corollary captures a notion of tamper detection similar to \cite{JW15}, whereas the following corollary recovers (and in fact improves the parameters of) the result by \cite{BK23}

\begin{corollary}[Unitary adversary form, informal]
    Let $\advU \subseteq \mathcal{U}_d(\CC)$ be a family of unitary operators such that $|\advU|\le 2^{d^{\alpha}}$ for some $\alpha<\tfrac12$ and for every $V\in\advU$,
   we have, 
      $\big|\Tr(V)\big|\le \sqrt{2}\,d^{3/4}.$

Then, for a Haar-random unitary $U\in\mathcal U_d(\CC)$,
\[
  \Pr_U\!\big[ (\dec_U,\enc_U)\ \text{is }\varepsilon\text{-tamper secure against }\advU\big]
  \;\ge\; 1-\operatorname{negl}(k),
\]
with $\varepsilon=\operatorname{negl}(k)$ and $\gamma=\mathcal O(1)$. In particular, for any fixed $\alpha<\tfrac12$ and sufficiently large $d$ there exists a (fixed) unitary $U$ such that $(\dec_U,\enc_U)$ is $\varepsilon$-tamper secure against $\advU$.
\end{corollary}

Moreover, we extend the tamper detection result to the case when the set of plaintext $\mathcal{M}$ consists of quantum states.
The result is formalized and presented in Theorem~\ref{thm:tamper_detection_quantum}. 
The encoder for this case remains the same as that of the classical plaintext, however, the decoder is modified to a two-valued PVM (instead of computational basis measurement), in order to preserve the superposition and entangelement in the plaintext.

\fi
\subsubsection*{Techniques and proof sketch}

Our proofs
\ifExtAbs
\else
  for Theorem~\ref{thm:tamper_detection_code} and Theorem~\ref{thm:tamper_detection_quantum}
\fi
 are based on probabilistic analysis.
In particular, we give an encoding-decoding scheme $(\enc_U-\dec_U)$ for every unitary $U$ (of appropriate dimension).
Suppose we say that a unitary operator $U$ is \emph{good against an adversarial family $\advCPTP$} if $(\enc_U-\dec_U)$ is an $\eps$-tamper  detection code.\footnote{we will typically consider $\eps$ to be exponentially small in the plaintext length.} 
We show that a randomly drawn $U$ is \emph{good} with high probability, and hence, there must exist a \emph{good} $U$ against $\advCPTP$.

\begin{figure}[t]
\centering
\begin{tcolorbox}[colback=white]
\begin{enumerate}
    \item Sample $U\sim$ Haar from $\mathcal{U}_d(\mathbb{C})$.
    \item \textbf{Encoding:}  
    Input $m\in\{0,1\}^k$. Prepare $\ket{m}^A\otimes\ket{0}^B$, where $|B|=\log(d)-k$. 
    \[
    \enc_U(m) \coloneqq U \big( \, \ket{m}^A\otimes\ket{0}^B \big).
    \]
    \item \textbf{Decoding:}  
    On $\rho$, apply $U^\ast$ and measure $B$ to obtain an outcome $b$:
    \begin{itemize}
        \item If $b\neq 0^{|B|}$, output $\perp$.
        \item If $b=0^{|B|}$, measure $A$ in the computational basis and output the result.
    \end{itemize}
\end{enumerate}

\end{tcolorbox}
\caption{Description of the Haar-random encoding scheme}
\label{fig:Haar Encoding}
\end{figure}

At a high level, the proof shows that for a fixed tampering map $\Phi$ with low entanglement fidelity and low minimal Kraus rank, a Haar-random encoding unitary $U$ (\cref{fig:Haar Encoding}) makes the probability that $\Phi$ maps one codeword into another (or into the decoding subspace corresponding to any valid message) negligibly small.  The argument proceeds along three main steps: 
\ifExtAbs
\else
\begin{enumerate}
    \item {\bf Reduction to pairwise overlaps.}  Using the definition of the Haar-random encoding and the projective decoder, we reduce the bad event (decoder does not output $\perp$ after tampering) to a small collection of overlap events between the images of computational basis codewords under $U$ and the image of those vectors after applying $\Phi$.
    \item {\bf Moment / concentration bounds for random unitaries.}  For a fixed pair of vectors and a fixed (low-rank) channel, we estimate the moments (and obtain tail bounds) of the corresponding overlap random variable when $U$ is Haar distributed. These estimates exploit standard facts about Haar measure and moments of unitary matrix entries, together with careful accounting for the (minimal) Kraus decomposition of $\Phi$. 
    Intuitively, the rank condition helps us to bound the moments for vector pair $(s,t)$ when $s \neq t$. 
    This is the event that the encoder encoded message $s$ but the adversary managed to trick the decoder into decoding $t$.
  Furthermore, the entanglement fidelity condition helps in bounding moments for $s=t$ case.
  \item  {\bf Union bounds:} Then, to acheive security for all plaintexts against all the adversarial maps, we union-bound over all messages $m \in \mathcal{M}$ and over all maps $\Phi \in \advCPTP$.

  The encoding-decoding map given above can be extended to the case when the plaintexts $\mathcal{M}$ are themselves quantum, however
  the proof for existence of tamper detection code for quantum messages requires an additional step. 
    \item {\bf Net argument.}   A fine enough $\varepsilon$-net (on the unit sphere / projectors) together with the preceding concentration estimates yield that the probability a fixed $\Phi$ causes any decoding error is negligibly small; taking a union bound over the family size $|\advCPTP| \le 2^{d^\alpha}$ preserves negligibility under the theorem’s parameter constraints.
\end{enumerate}
\fi

\subsubsection*{Contribution Part B: Possibility of an universal relaxed tamper detection}
A central obstacle in tamper detection is that certain adversarial families are inherently unavoidable.
In the classical setting, \cite{JW15} showed that the two problematic families are the identity function and the constant functions.
The former can safely be ignored: identity tampering is undetectable by definition, but also harmless, since the message is preserved.
This motivates the notion of relaxed tamper detection (RTD)~\cite{BK23}, which requires the decoder either to abort or, to correctly recover the original message.

Under this relaxed notion, however, constant functions still constitute a fatal obstruction:
\cite{JW15} proved that no (classical) relaxed tamper-detection code can protect against the family of all constant functions, which has size $2^n$.
In contrast, we show that quantum encodings can protect against this very same family.
Specifically, we design a unitary encoding–decoding pair that achieves negligible error against all classical constant functions, provided the codeword length, $n$, is a sufficiently large multiple of the message length~$k$.
In fact, our result holds in a more general setting, 
\ifExtAbs
\else
and the full formulation is deferred to the technical discussion in \Cref{sec:universal_tamper}.
\fi

This separation suggests a fundamental difference between classical and quantum tamper detection:
whereas constant functions are unconditionally impossible to handle classically, they become benign in the quantum setting.
Motivated by this, we conjecture that quantumly one can in fact achieve relaxed tamper detection against \emph{any} family of CPTP maps of size up to $2^{d^{\alpha n}}$, where $\alpha <1$.
As evidence, we establish this conjecture for the special case of constant replacement channels, showing that Haar-random encodings achieve security with negligible error.
Recall that classically, constant channels were the only bad family for relaxed tamper detection.
Hence, ruling out replacement channels, which are generalizations of constant functions, provides an indication that there are no bad families for tamper detection with quantum encodings. 
However, we could not prove this result and leave this as a conjecture\ifExtAbs
\else
(see~Conjecture~\ref{con:tamper_dection_against_any_poly}).

Nonetheless, our results provide the first proof that quantum tamper detection is strictly more powerful than its classical counterpart.
\fi

\ifExtAbs
\else

\subsection{Other related coding theoretic and cryptographic primitives}

Perhaps the closest coding-theoretic object that is related to tamper detection are non-malleable codes. 
This concept was introduced in~\cite{DPW18}: in this case, a decoder must  output the original message $m$, or output a message that is unrelated to $m$.
Over the years, non-malleable codes have lead to a well established field of non-malleable cryptography~\cite{CG16,FHMV17,BDSKM18,DPW18,AO20,AKO+22}, to name a few.
Only recently, non-malleable cryptography has found its way in the quantum computing domain~\cite{BJK21arxiv, BB23arxiv, ABJ24, BGJR25}.
We note that both tamper detection codes and relaxed tamper detection codes are stronger objects than non-malleable codes. That is, any (relaxed) tamper detection scheme $(\enc-\dec)$ against an adversarial family $\advCPTP$ is also a non-malleable code against against $\advCPTP$.  
As far as applications are considered, one is often interested in getting non-malleable security against a class that is much broader than the adversarial class of tamper detection. 
In particular, one wants non-malleable guarantees against classes of functions that include the identity function as well as constant functions. 
To achieve this, \cite{JW15} compose a tamper detection code with a \emph{leakage-resilient code} which yields a non-malleable code.
Interestingly, since our quantum encoding already provides relaxed tamper detection guarantees against constant functions (and their generalizations, \emph{i.e.}, replacement channels), we suspect that one might not need an additional layer of leakage-resilience to get non-malleable codes out of quantum encodings.
Indeed, if our Conjecture~\ref{con:tamper_dection_against_any_poly} is true, then we already have a non-malleable code against a very broad family.

Some other notions of non-malleability have also been considered.  
For example~\cite{ABW09,AM17} define  non-malleable codes in a \emph{keyed} setting, which seems closer to a quantum analogue of message authentication codes (MAC). 
It is known that in a keyed setting, MAC does provide non-malleable guarantees (see, for example,~\cite{BW16}). 
The motivation as well as proof technique in a keyed model tend to differ drastically (from those that are keyless) simply because of the assumption of a key that the adversary can neither know nor tamper with.
In particular, such schemes can exploit entropies, uncertainty principles and related combinatorial objects such as randomness extractors.
Keyless primitives have no such information that is hidden from the adversary, which makes it impossible to use such objects and proof techniques.
As mentioned before, a recent line of work has initiated the study of non-malleable codes~\cite{BJK21arxiv, BB23arxiv, ABJ24, BGJR25} in the \emph{split-state} model, where multiple separated and non-communicating adversaries (who are potentially pre-entangled)  tamper separate parts of the codeword but do so independently. 
Such models are keyless but due to the \emph{split} nature of the adversary, it seems the results that such model can provide are incomparable with our model. 
To be more precise, the split adversary model constructs efficient codes against a restricted class of adversaries whereas the model we consider, similar to~\cite{JW15}, aims to provide existential results against adversaries that are not bounded by such a split-action condition.

 On codes that are more related to tamper detection,  Pauli Manipulation Detection (PMD) codes \cite{Ber24} are a subclass of tamper detection codes where  each encoded qubit may be tampered with independently.
 \cite{BK23} introduced these as a natural analogue of (classical) AMD codes and call them quantum AMD codes.
 However, they restrict their analysis to the case when the plaintexts are classical messages. \cite{Ber24} considered a fully quantum version of this. 
 Remarkably, by composing PMD codes with stabilizer codes and classical non-malleable codes, \cite{Ber24} obtains high-rate, near-optimal constructions: approximate quantum erasure codes approaching the Singleton bound and keyless  qubit-wise authentication — a task provably impossible in the classical setting.

Quantum error-correction codes (QEC) (or error detection codes) and quantum message authentication codes  are two primitives that deserve a discussion here as they also deal with the task of detection of errors in message transmission.

Most QEC analyses assume {\em weight-} or {\em locality-bounded} errors (for example, that at most $t$ physical qubits are corrupted), and a code of distance $d_{\mathrm{code}}$ can detect arbitrary errors of weight up to $d_{\mathrm{code}}-1$~\cite{NC00}.
This locality/weight hypothesis underlies explicit, efficient code constructions and efficient fault-tolerant protocols, but it excludes powerful adversaries that may apply global, highly entangling CPTP maps or concentrate action on large subspaces.
In contrast to this, tamper detection codes aim to detect errors even of weight much higher than the distance of code, but at the cost of probabilistic encoding - decoding, and in some cases, even non-explicit constructions.
It should also be noted that QEC can \emph{correct} errors of weight up to 
$\lfloor \tfrac{d_{\mathrm{code}}-1}{2} \rfloor$, but there is no \emph{correction} required from a tamper detection code.
They merely \emph{detect} an error, which in principle, could be of unbounded weight.

Quantum message authentication  codes (Q-MAC)~\cite{BCG+02} are explicitly designed codes that detect adversarial tampering, but they typically rely on secret keys shared between sender and receiver, and often aim to give both integrity and (sometimes) confidentiality.
By contrast, tamper detection codes are  {\em public and keyless}.
The tamper-detection model asks a different question: Suppose a public encoder/decoder family has access only to public classical randomness but no private randomness (\emph{i.e.}, no shared private key), what structural constraints on the adversarial family allow detection of tampering?
While there is technical overlap in tools (random/unitary encodings, projector decoders, concentration of measure), the resource and threat-model differences mean that results and constructions from QEC or Q-MAC do not directly carry over to tamper detection setting.

We also point out why key distribution schemes (such as QKD~\cite{BB84}) do not directly subsume the public, keyless tamper-detection problem. 
First, the objectives differ: QKD aims to establish a shared secret key between parties (typically over an insecure channel) and its security statements concern key correctness and secrecy, whereas tamper detection concerns integrity of a stored quantum message (recover vs.\ output the detection symbol $\bot$). 
Moreover, one may question, why not first distribute keys using QKD and then use these keys to establish Q-MAC, which effectively, gives tamper detection.  
Note that QKD protocols require an authenticated classical channel (or an initial shared key, often referred to as the weak shared password) to prevent man-in-the-middle attacks; thus QKD’s guarantees are not delivered in a purely public, keyless resource model.
Additionally, QKD detects adversaries via interactive, statistical parameter estimation (basis choices, sample testing and error rates) on transmitted qubits; these tests are destructive and probabilistic, and are designed for the transmission/communication setting rather than for a single-shot test that preserves a stored quantum state.

\subsection{Limitations and open problems}

The Haar-random scheme is inherently  {infinite and inefficient; while this is already more constructive compared to many results in coding theory that are probabilistic in nature,} 
finding efficient or efficiently sampleable constructions that achieve comparable parameters is an important open problem. 
However, note that such a construction is not known even against a classical adversary.  Related directions left for future work include:
\begin{itemize}
    \item The classical theory of tamper detection~\cite{JW15} shows the existence of tamper detection code families of size up to $2^{2^{\alpha n}}$ for any $\alpha <1$.
    ~\cite{BK23} achieve similar results against unitary operators, but with $\alpha <\frac{1}{4}$. 
    They conjectured that similar to the classical theory, one should get tamper detection codes for any $\alpha <1$.
    In this work, we manage to raise the ceiling from $\alpha < \frac{1}{4}$ to $\alpha < \frac{1}{2}$ but it would be interesting to see if one can actually push the $\alpha$ even closer to 1.
    \item Of course, Conjecture~\ref{con:tamper_dection_against_any_poly} on universal quantum tamper detection remains one of the main targets.
    Recall that with classical encodings, one had to compose tamper detection codes (or a related notion called the \emph{predictable codes}) with leakage resiliant codes to get non-malleable codes. 
    If the conjecture is true, then one can get non-malleable codes against quantum adversaries, without any such composition. 
    \item Moving away from the existential results, it would be interesting to see an explicit construction, even for a family of size up to $2^{\poly(n)}$. 
    Currently, the only explicit results known are against the family with size $2^{\mathcal{O}(n)}$  (PMD codes~\cite{Ber24} and quantum AMD codes~\cite{BK23} provide security against Pauli channels).
    \item   We only considered non-adaptive single-shot tampering. Extending the analysis to adaptive or multiple tampering attempts, to leakage before or during tampering, or to stronger security goals such as \emph{no-cloning} resistance are natural next steps.

\end{itemize}

\fi

\ifExtAbs
\else

\section{Preliminaries}

We denote by $\NN$ the set of natural numbers, and by $[n]$ the set $\{1, \dots, n\}$. The notation $\log_b(\cdot)$ refers to the logarithm with base $b$. In particular, we write $\log(\cdot)$ for the base-$2$ logarithm and $\ln(\cdot)$ for the natural logarithm (base $e$).

Throughout this work, all Hilbert spaces are assumed to be finite-dimensional. For Hilbert spaces $\cH_A$ and $\cH_B$, we denote by $\mathcal{B}(\cH_A, \cH_B)$ the space of bounded linear operators from $\cH_A$ to $\cH_B$. When $\cH_A = \cH_B = \cH$, we simply write $\mathcal{B}(\cH)$. The $d \times d$ identity operator on $\cH \simeq \CC^d$ is denoted by $I_d$. The trace over $\cH$ is denoted $\Tr[\cdot]$. An operator $M \in \mathcal{B}(\cH)$ is said to be positive semi-definite if and only if it is Hermitian (\ie, $M^* = M$) with non-negative real eigenvalues. In such a case, we write $M \succeq 0$. For any $p \geq 1$, the Schatten $p$-norm is denoted $\|\cdot\|p$, while $\|\cdot\|\infty$ denotes the operator norm.

A Hilbert space $\cH$ will be referred to as a quantum system. A quantum state (or density operator) on $\cH$ is a positive semi-definite operator $\rho \in \mathcal{B}(\cH)$ with unit trace. The set of density operators is convex, and its extreme points are precisely the rank-one projectors $\ketbra{\psi}{\psi}$, where $\ket{\psi} \in \cH$ is a unit vector, referred to as a pure state. We will use both vector and matrix notations interchangeably for pure states.

For a composite system $\cH \simeq \cH_{A_1} \otimes \cdots \otimes \cH_{A_n}$ with $n$ parties, the partial trace over the $i$-th subsystem is denoted $\Tr_{A_i}[\cdot]$. When subsystems have different dimensions, we also index the partial trace by the dimension being traced out. For example, if $\cH \simeq \CC^d \otimes \CC^D$, the notation $\Tr_D[\cdot]$ denotes the partial trace over the second tensor factor.

A linear map $\Phi: \mathcal{B}(\cH) \to \mathcal{B}(\cH)$ is said to be completely positive if, for every $k \in \NN$, the extended map $\Phi \otimes \id_k$ is positive (\ie, it maps positive operators to positive operators), where $\mathrm{id}_k$ denotes the identity map on $\mathcal{B}(\CC^k)$. It is Trace Preserving if, for all $\rho \in \mathcal{B}(\cH)$, one has $\Tr[\Phi(\rho)] = \Tr[\rho]$. A quantum channel is a linear map that is both Completely Positive and Trace Preserving (\CPTP).

A Positive Operator-Valued Measure (\POVM) on $\cH$ is a finite collection of positive semi-definite operators~${M_i}$ such that $\sum_i M_i = I_d$. A Projection-Valued Measure (\PVM) is a \POVM in which each $M_i$ is an orthogonal projection.

We denote by $\mathfrak{S}_n$ the symmetric group on $n$ elements. The group of $d \times d$ unitary matrices is denoted~$\mathcal{U}_d(\CC)$, consisting of all matrices $U$ satisfying $UU^* = U^*U = I_d$. Some properties of $\mathfrak{S}_n$ and $\mathcal{U}_d(\CC)$ relevant to this work are collected in the \Cref{app:symmetric_group,app:weingarten_calculus}. An operator $V \in \mathcal{B}(\CC^d, \CC^D)$ with $D \geq d$ is called an isometry if $V^* V = I_d$. In this case, the operator $V V^* \in \mathcal{B}(\CC^D)$ is the orthogonal projection onto the image of $V$. A unitary is a special case of an isometry when $D=d$.

\subsection{The structure of \texorpdfstring{\CPTP}{CPTP} maps}

By the \emph{Kraus representation} \cite[Corollary 2.27]{Wat18}, any \CPTP map $\Phi: \mathcal{B}(\CC^d) \to \mathcal{B}(\CC^D)$ admits a decomposition of the form
\begin{equation*}
    \Phi(X) = \sum^r_{i=1} K_i X K_i^*,
\end{equation*}
where the integer $r$ is called the \emph{Kraus rank}, and the operators $\{K_i\} \subset \mathcal{B}(\CC^d, \CC^D)$ are called the \emph{Kraus operators} associated with $\Phi$ and satisfy the normalization condition $\sum_i K_i^* K_i = I_d$. This representation is not unique: distinct sets of Kraus operators can define the same quantum channel. The smallest integer $r$ for which such a representation exists is called the \emph{minimum Kraus rank} of $\Phi$, and this value is uniquely determined by the channel.

The \emph{Choi–Jamiołkowski representation} associates to any linear map $\Phi: \mathcal{B}(\CC^d) \to \mathcal{B}(\CC^D)$ its \emph{Choi matrix} $J(\Phi) \in \mathcal{B}(\CC^d, \CC^D)$, defined by
\begin{equation*}
    J(\Phi) \coloneq (\Phi \otimes \id_d) \sum_{i,j=1}^d \ketbra{ii}{jj}.
\end{equation*}
The Choi matrix $J(\Phi)$ fully characterizes the map $\Phi$ and yields a criterion for \CPTP, known as the \emph{Choi theorem} \cite[Corollary 2.27]{Wat18}: a linear map $\Phi : \mathcal{B}(\CC^d) \to \mathcal{B}(\CC^D)$ is \CPTP if and only if $J(\Phi) \succeq 0$ and satisfies the trace condition $\Tr_D[J(\Phi)] = I_d$.

Moreover, the Choi matrix provides a direct characterization of the minimum Kraus rank \cite[Corollary 2.27]{Wat18}: the minimum Kraus rank of $\Phi$ is precisely the matrix rank of $J(\Phi)$. For a collection $\mathcal{F}$ of \CPTP maps, we define
\begin{equation*}
    \rank(\mathcal{F}) \coloneqq \max_{\Phi \in \mathcal{F}} \rank(J(\Phi)),
\end{equation*}
the largest minimum Kraus rank among the maps in $\mathcal{F}$.

The \emph{Stinespring representation} \cite[Corollary 2.27]{Wat18} provides an alternative, yet equivalent, characterization of \CPTP maps in terms of an isometric embedding into a larger Hilbert space. Specifically, any \CPTP map $\Phi: \mathcal{B}(\CC^d) \to \mathcal{B}(\CC^D)$ can be written as  
\begin{equation*}
    \Phi(X) = \Tr_{\mathcal{E}}\big[ V X V^* \big],
\end{equation*}
where $\mathcal{E}$ denotes an auxiliary Hilbert space (the \emph{environment}), and $V: \CC^d \to \CC^D \otimes \mathcal{E}$ is an isometry satisfying $V^* V = I_d$. This representation is unique up to isometries acting on the environment space. The dimension of $\mathcal{E}$ can always be chosen to match the minimum Kraus rank of $\Phi$, thereby linking the Stinespring, Kraus, and Choi representations into a unified structural description of quantum channels.

The \emph{entanglement fidelity} of a quantum channel $\Phi: \mathcal{B}(\CC^d) \to \mathcal{B}(\CC^d)$ quantifies how well the channel preserves entanglement with a reference system. Let $\ket{\Omega} \coloneqq \frac{1}{\sqrt{d}} \sum_{i=1}^d \ket{ii}$ denote the \emph{maximally entangled state}. The entanglement fidelity is defined by \cite[Definition 3.30]{Wat18}:
\begin{equation*}
    F_e(\Phi) \coloneqq \Big\langle \Omega \Big| (\Phi \otimes \id_d) \big( \ketbra{\Omega}{\Omega} \big) \Big| \Omega \Big\rangle.
\end{equation*}
If $\Phi$ admits a Kraus decomposition with Kraus operators $\{K_i\}$, then the entanglement fidelity admits the equivalent expression \cite[Proposition 3.31]{Wat18}:
\begin{equation*}
    F_e(\Phi) = \frac{1}{d^2} \sum_k |\Tr[K_i]|^2.
\end{equation*}
Both definitions are equivalent and independent of the chosen Kraus representation. For a collection $\mathcal{F}$ of \CPTP maps, we define
\begin{equation*}
    F_e(\mathcal{F}) \coloneqq \max_{\Phi \in \mathcal{F}} F_e(\Phi),
\end{equation*}
the largest entanglement fidelity among the maps in $\mathcal{F}$.

\subsection{Tamper detection}

Let $\mathcal{M}$ denote a finite set of classical messages, and let $\CC^d$ be a Hilbert space of dimension $d$. A quantum public and keyless code encodes a classical message $m \in \mathcal{M}$ into a quantum codeword in $\cH$.

\begin{definition}[Quantum Encoding-Decoding Scheme] \label{def:quantum_encoding_decoding_scheme}
    A \emph{quantum encoding–decoding scheme} for $\mathcal{M}$ is defined as a pair of quantum algorithms $(\enc,\dec)$ such that $\enc : \mathcal{M} \to \CC^d$ and $\dec : \CC^d \to \mathcal{M} \cup \{\perp\}$, satisfying the following correctness condition:
    \begin{equation*}    
        \forall m \in \mathcal{M},\quad \Pr[\dec(\enc(m)) = m] = 1.
    \end{equation*}
    Here, the probability is taken over the internal randomness of both the encoder and the decoder.
\end{definition}
{Typically, two cases are of interest:
When $\mathcal{M}$ is the set of \emph{classical} messages, $\mathcal{M} = \{0,1\}^k$ and when $\mathcal{M}$ is the set of quantum states, $\mathcal{M} = \lbrace \psi \in \mathbb{C}^{2^k}, \Vert \psi \Vert_2 = 1 \rbrace$.}

{To specify the dimension of such a scheme, we refer to this as $\left(\log{\mathcal{M}}, \log d\right)$ encoding-decoding scheme. 
The \emph{rate} of the scheme is deseribed by $\textrm{R} \coloneqq \frac{\log d}{\dim{M}}$.}

We consider an adversarial model in which the codeword may be tampered via a map $f \in \advCPTP$, where~$\advCPTP$ denotes a given \emph{family of tampering functions}.

\begin{definition}[Tamper Detection Against quantum Adversaries] \label{def:tamper_detection_against_quantum_adversaries}
    Let $\mathcal{M}$ be a finite set of messages. Let $\advCPTP$ be a set of \CPTP maps from $\mathcal{B}(\CC^d)$ to $\mathcal{B}(\CC^d)$.  
    An encoding–decoding scheme $(\enc, \dec)$ is said to be \emph{$\varepsilon$-tamper secure} against $\advCPTP$ if the following holds:
    \begin{equation*} 
        \forall m \in \mathcal{M}, \forall \Phi \in \advCPTP,\quad \Pr \left[ \dec \circ \Phi \circ \enc(m) = \perp \right] \geq 1 - \varepsilon.
    \end{equation*}
    As before, the probability is taken over the internal randomness of both the encoder and the decoder.
\end{definition}

\subsection{Haar random scheme}

We now introduce a quantum encoding–decoding scheme using unitaries sampled from the Haar measure.

\begin{definition}[Haar Random Scheme Family] \label{def:Haar_random_scheme_family}
    Let $\mathrm{HR}(d,k)$ denote the \emph{Haar random quantum encoding–decoding scheme family} with parameters $d$ and $k$, where $k \leq \log(d)$.
    Let $\mathcal{M} = \{0,1\}^k$ be the message set.
    Let $U \in \mathcal{U}_d(\CC)$ be a $d \times d$ unitary matrix sampled from the Haar measure over the unitary group, and define $(\dec_U, \enc_U)$ to be the \emph{Haar random quantum encoding–decoding scheme} associated with $U$. For each $m \in \mathcal{M}$, consider its canonical encoding as the computational basis state $\ket{m} \in \CC^d$. The pair $(\dec_U, \enc_U)$ is defined as follows:
    \begin{description}
        \item[Encoding:] Register $A$ will hold the message $m$ which will be padded on register $B$ with $\log(d)-k$ sized ancilla.
        \[ \enc_U(m) \coloneqq U \big( \ket{m}^A\otimes\ket{0}^{B} \big). \]
        \item[Decoding:] On receiving the state $\rho$, decoder first reverts the encoder using $U^{*}$, and then measures $B$ in computational basis. 
        Let the measurement result be denoted as $b$.
        \begin{itemize}
            \item If $b \neq 0^{\vert B \vert}$, then output $\widehat{m} = \perp$, indicating that there was a tampering.
            \item If $b =  0^{\vert B \vert}$, then measure $A$ in the computational basis and declare the output to be $\widehat{m}$.
        \end{itemize}
        The decoder can alternatively be described as follows: measure the received codeword using the \PVM  $\{ \Pi_1, \ldots, \Pi_{|\mathcal{M}|}, \Pi_\perp \}$, where $\Pi_m$ is the orthogonal projector onto $\enc_U(m)$, and $\Pi_\perp$ projects onto the orthogonal complement.
    \end{description}
\end{definition}

\begin{remark*}
    Correctness follows directly from the construction, and in particular, the orthogonality of the codewords.
\end{remark*}

The expansion factor for the tamper detection scheme is characterized by the size of ancillary qubits that are used.
More concretely, $\gamma \coloneq \frac{|B|}{|A|}$.
Ideally, we want schemes with $\gamma = \mathcal{O}(1)$, that is, the number if additional qubits used to be linear in the number of plaintext qubits. 

In \Cref{sec:tamper_detection}, we prove that the Haar-random scheme is $\varepsilon$-tamper secure against a family $\advCPTP$ of tampering maps, subject to constraints on the cardinality $|\advCPTP|$, the maximal minimal Kraus rank: $\mathrm{rank}(\advCPTP)$, and the maximal entanglement fidelity: $F_e(\advCPTP)$.
Note that although a Kraus decomposition of a map is not unique, the above quantities are invariant under the choice of Kraus decomposition.

\section{Tamper detection} \label{sec:tamper_detection}

In this section, we establish the main result of the paper: the existence of a tamper-detection code scheme that remains secure against a specified family of \CPTP maps.
 \begin{theorem}[Tamper detection for classical messages] \label{thm:tamper_detection_code}
    Let $\mathcal{M} = \{0,1\}^k$, and let $\advCPTP$ denote an adversarial family of \CPTP maps from $\mathcal{B}(\CC^d)$ to $\mathcal{B}(\CC^d)$ satisfying the following conditions, for parameters $0 \leq 2 \alpha < {\delta} < 1$:
    \begin{itemize}
        \item The cardinality of the adversarial family is bounded as $ \vert \advCPTP \vert \leq 2^{d^{\alpha}}$.
        \item The maximal minimal Kraus rank within the family satisfies $\rank(\advCPTP) \leq d^{1-\delta}$.
        \item The maximal entanglement fidelity of the family satisfies $F_e(\advCPTP) \leq \frac{2}{d^{\delta/2}}$.
    \end{itemize}
     Then, 
     \begin{equation*}
         \Pr_U \big[ (\dec_U, \enc_U) \text{ is $\varepsilon$-tamper secure against $\advCPTP$} \big] \geq 1 - \operatorname{negl}(k),
     \end{equation*}
     with $\varepsilon = \operatorname{negl}(k)$ and $\gamma= \mathcal{O}(1)$. In particular, there exists a unitary $U \in \mathcal{U}_d(\CC)$ such that $(\dec_U, \enc_U)$ is $\varepsilon$-tamper secure against $\advCPTP$.
\end{theorem}

The proof of \Cref{thm:tamper_detection_code}, presented in \Cref{sec:tamper_detection_code_proof}, relies on several intermediate lemmas established below.

\subsection{Intermediate lemmas}

Let $U \in \mathcal{U}_d(\CC)$ be a unitary matrix and define $\ket{\psi_m} \coloneqq U \ket{m}$ for all messages $m \in \mathcal{M}$. Consider a \CPTP map $\Phi: \mathcal{B}(\CC^d) \to \mathcal{B}(\CC^d)$ with Kraus decomposition:
\begin{equation*}
    \Phi(X) = \sum_{i=1}^r K_i X K_i^*.
\end{equation*}
This representation is generally non-unique, nevertheless the quantities introduced in this section will be invariant under the choice of Kraus operators.
Define the random variable $X_{st}$ as
\begin{align*}
    X_{st} &\coloneqq \Tr \Big[ \ketbra{\psi_t}{\psi_t} \cdot \Phi\big( \ketbra{\psi_s}{\psi_s} \big) \Big] \\
    &= \bra{\psi_t} \Phi\big( \ketbra{\psi_s}{\psi_s} \big) \ket{\psi_t} \\
    &= \sum_{i=1}^r \bra{\psi_t} K_i \ketbra{\psi_s}{\psi_s} K_i^* \ket{\psi_t} \\
    &= \sum_{i=1}^r \bra{\psi_t} K_i \ket{\psi_s} \otimes \bra{\psi_s} K_i^* \ket{\psi_t} \\
    &= \sum_{i=1}^r \bra{\psi_t}\bra{\psi_s} \big( K_i \otimes K_i^* \big) \ket{\psi_s} \ket{\psi_t} \\ 
    &= \sum_{i=1}^r \bra{t}\bra{s} \big( U^* \otimes U^* \big) \big( K_i \otimes K_i^* \big) \big( U \otimes U \big) \ket{s} \ket{t}.
\end{align*}

We now establish the following lemma for the first moment of the operators $X_{st}$. Recall that for any \CPTP map, the entanglement fidelity remains invariant under different choices of Kraus representations. As a direct consequence, the quantity
\begin{equation*}
    d^2 \cdot F_e(\Phi) = \sum^r_{i=1} {\vert \Tr(K_i) \vert}^2,
\end{equation*}
is independent of the specific Kraus decomposition $\{K_i\}_{i=1}^r$ representing the \CPTP map.
\begin{lemma} \label{lem:first_moment_upperbound}
    Let $\phi \in [0,1]$ and $\Phi: \mathcal{B}(\CC^d) \to \mathcal{B}(\CC^d)$ a \CPTP map such that $F_e(\Phi) \leq \phi^2$. Then
    \begin{equation*}
        \E_U \big[ X_{st} \big] \leq
        \begin{cases}
            \frac{\phi^2 d^2 + d}{d^2+d} &\text{if $s = t$} \\[0.5em]
            \frac{d}{d^2-1} &\text{if $s \neq t$} \\
        \end{cases}
    \end{equation*}
\end{lemma}
\begin{proof}
    From the definition of $X_{st}$,
    \begin{equation*}
        \E_U \big[ X_{st} \big] = \sum_{i=1}^r \bra{t}\bra{s} \E_U \Big[ \big( U^* \otimes U^* \big) \big( K_i \otimes K_i^* \big) \big( U \otimes U \big) \Big] \ket{s} \ket{t}.
    \end{equation*}
    Using \Cref{lem:Weingarten_calculus_first_moments}, we have
    \begin{equation*}
        \E_U \Big[ \big( U^* \otimes U^* \big) \big( K_i \otimes K_i^* \big) \big( U \otimes U \big) \Big] = c_I \cdot I_{d^2} + c_F \cdot F,
    \end{equation*}
    with
    \begin{equation*}
        c_I = \frac{\Tr \big[ K_i \otimes K_i^* \big] - \frac{1}{d} \Tr \Big[ \big( K_i \otimes K_i^* \big) \cdot F \Big]}{d^2 - 1} \qquad \text{and} \qquad c_F = \frac{\Tr \Big[ \big( K_i \otimes K_i^* \big) \cdot F \Big] - \frac{1}{d} \Tr \big[ K_i \otimes K_i^* \big]}{d^2 - 1}.
    \end{equation*}
    Using the swap trick \cref{eq:swap_trick}: $\Tr[(A \otimes B) \cdot F] = \Tr [A \cdot B]$, we obtain
    \begin{equation*}
        c_I = \frac{{\big| \Tr[K_i] \big|}^2 - \frac{1}{d} \Tr \big[ K_i^* K_i \big]}{d^2 - 1} \qquad \text{and} \qquad c_F = \frac{\Tr \big[ K_i^* K_i \big] - \frac{1}{d} {\big| \Tr[K_i] \big|}^2}{d^2 - 1}.
    \end{equation*}
    Thus
    \begin{equation*}
        \E_U \big[ X_{st} \big] = \sum_{i=1}^r c_I \cdot \bra{t} \bra{s} I_{d^2} \ket{s} \ket{t} + c_F \cdot \bra{t} \bra{s} F \ket{s} \ket{t}.
    \end{equation*}
    If $s = t$, then
    \begin{equation*}
        \E_U \big[ X_{st} \big] = \sum_{i=1}^r c_I + c_F = \sum_{i=1}^r \frac{{\big| \Tr[K_i] \big|}^2 + \Tr \big[ K_i^* K_i \big]}{d^2 + d} \leq \frac{\phi^2 d^2 + d}{d^2+d}.
    \end{equation*}
    If $s \neq t$, then
    \begin{equation*}
        \E_U \big[ X_{st} \big] = \sum_{i=1}^r c_F = \sum_{i=1}^r \frac{\Tr \big[ K_i^* K_i \big] - \frac{1}{d} {\big| \Tr[K_i] \big|}^2}{d^2 - 1} \leq \frac{d}{d^2-1}.
    \end{equation*}
    Where we use the fact that $\sum_i \Tr [K_i ^* K_i] = \Tr[I_d]$.
\end{proof}

We now establish the following lemma for the higher-order moments of the operators $X_{st}$. This builds upon the previously-used invariant entanglement fidelity as well as the minimum Kraus rank.

\begin{lemma} \label{lem:nth_moment_upperbound}
    Let $\phi \in [0,1]$  and $\Phi: \mathcal{B}(\CC^d) \to \mathcal{B}(\CC^d)$ a \CPTP map with minimum Kraus rank $r$ and such that $F_e(\Phi) \leq \phi^2$. Then
    \begin{equation*}
        \E_U \big[ (X_{st})^n \big] \leq
        \begin{cases}
             \Big( 1 + \mathcal{O} \big( \tfrac{4n^2}{d} \big) \Big) {\Big( \sqrt{r} \frac{\phi^{2} d + 2n}{\phi d} \Big)}^{2n} &\text{if $s = t$} \\[0.5em]
             \Big( 1 + \mathcal{O} \big( \tfrac{4n^2}{d} \big) \Big) { \Big( \frac{r n (d + n)}{d^2} \Big)}^n &\text{if $s \neq t$} \\
        \end{cases}
    \end{equation*}
\end{lemma}
\begin{proof}
    From the definition of $(X_{st})^n$,
    \begin{align*}
        (X_{st})^n &= \sum_{i \in [r]^n} {\big| \bra{\psi_t} K_{i_1} \ket{\psi_s} \big|}^2 \otimes \cdots \otimes {\big| \bra{\psi_t} K_{i_n} \ket{\psi_s} \big|}^2 \\
        &= \sum_{i \in [r]^n} \bra{\psi_t}\bra{\psi_s} \big( K_{i_1} \otimes K_{i_1}^* \big) \ket{\psi_s} \ket{\psi_t} \otimes \cdots \otimes \bra{\psi_t}\bra{\psi_s} \big( K_{i_n} \otimes K_{i_n}^* \big) \ket{\psi_s} \ket{\psi_t} \\
        &= \sum_{i \in [r]^n} \big( \bra{t}\bra{s} \big)^{\otimes n} \cdot {U^*}^{\otimes 2n} \cdot \Big[ \big( K_{i_1} \otimes K_{i_1}^* \big) \otimes \cdots \otimes \big( K_{i_n} \otimes K_{i_n}^* \big) \Big] \cdot U^{\otimes 2n} \cdot \big( \ket{s}\ket{t} \big)^{\otimes n}.
    \end{align*}
    Using \Cref{lem:weingarten_calculus_bound}, we have
    \begin{equation*}
        \E_U \bigg[ {U^*}^{\otimes 2n} \cdot \Big[ \underbrace{\big( K_{i_1} \otimes K_{i_1}^* \big) \otimes \cdots \otimes \big( K_{i_n} \otimes K_{i_n}^* \big)}_{\coloneqq X} \Big] \cdot U^{\otimes 2n} \bigg] \preceq \Big( 1 + \mathcal{O} \big( \tfrac{4n^2}{d} \big) \Big) \frac{1}{d^{2n}} \sum_{\pi \in \mathfrak{S}_{2n}} \Tr \big[ V(\pi)^{\shortminus 1} X \big] \cdot V(\pi).
    \end{equation*}
    To compute $\Tr \big[ V(\pi)^{\shortminus 1} X \big]$ we will use the generalized swap trick \cref{eq:generalized_swap_trick}:
    \begin{equation*}
        \Tr \big[ V(\pi)^{\shortminus 1} (M_1 \otimes \cdots \otimes M_k) \big ] = \prod^{\ell}_{j=1} \Tr \big[ M_{c_j(1)} \cdots M_{c_j(\ell_j)} \big],
    \end{equation*}
    where $c_j(k)$ denotes the $k$-th element of the $j$-th cycle and $\ell_j$ denotes its length. Let the operator $\tilde{K}_{i_j}$ be defined by $K_{i_{\lceil j/2 \rceil}}$ if $j$ is odd and $K^*_{i_{\lfloor j/2 \rfloor}}$ is $j$ is even. Then the tensor product $X$ becomes
    \begin{equation*}
        X = \big( K_{i_1} \otimes K_{i_1}^* \big) \otimes \cdots \otimes \big( K_{i_n} \otimes K_{i_n}^* \big) = \big( \tilde{K}_{i_1} \otimes \tilde{K}_{i_2} \big) \otimes \cdots \otimes \big( \tilde{K}_{i_{2n-1}} \otimes \tilde{K}_{i_{2n}} \big),
    \end{equation*}
    and the trace $\Tr[V(\pi)^{\shortminus 1} X]$ becomes
    \begin{equation*}
        \Tr \big[ V(\pi)^{\shortminus 1} X \big] = \prod^{\ell}_{j=1} \Tr \big[ \tilde{K}_{i_{c_j(1)}} \cdots \tilde{K}_{i_{c_j(\ell_j)}} \big].
    \end{equation*}
    Each term of the product can be bounded by $d$ as
    for all operators $M$, the inequality $\Tr[M] \leq d \cdot \norm{M}_{\infty}$ holds, and $\norm{K}_{\infty} \leq 1$ for all Kraus operators. Then
    \begin{equation*}
        \Tr \big[ V(\pi)^{\shortminus 1} X \big] \leq d^{\# \pi},
    \end{equation*}
    where $\# \pi \coloneqq \ell$ is the number of cycle in the cycle decomposition $\pi = (c_1, \ldots, c_\ell)$. We now divide the proof into two cases: $s \neq t$ and $s = t$.

    If $s \neq t$, then the scalar quantity $(\bra{t}\bra{s})^{\otimes n} \cdot V(\pi) \cdot (\ket{s}\ket{t})^{\otimes n}$ is non-zero and equal to $1$ if and only if the permutation $\pi$ satisfies $i \oplus \pi(i) = 1$ for all $i$, \ie the permutations maps even indices to odd indices and odd indices to even indices. Let $\tilde{\mathfrak{S}}_{2n}$ be the subset of such permutation, that is
    \begin{equation*}
        \tilde{\mathfrak{S}}_{2n} \coloneqq \set[\big]{ \pi \in \mathfrak{S}_{2n} \;:\; i \oplus \pi(i) = 1}.
    \end{equation*}
    Then from \Cref{lem:parity_alternating_permutation_property}
    \begin{align*}
        \E_U \big[ (X_{st})^n \big] &\leq  \sum_{i \in [r]^n} \Big( 1 + \mathcal{O} \big( \tfrac{4n^2}{d} \big) \Big) \frac{1}{d^{2n}} \sum_{\pi \in \tilde{\mathfrak{S}}_{2n}} \prod^{\ell}_{j=1} \Tr \big[ \tilde{K}_{i_{c_j(1)}} \cdots \tilde{K}_{i_{c_j(\ell_j)}} \big] \\
        &\leq  \sum_{i \in [r]^n} \Big( 1 + \mathcal{O} \big( \tfrac{4n^2}{d} \big) \Big) \frac{1}{d^{2n}} \sum_{\pi \in \tilde{\mathfrak{S}}_{2n}} d^{\#\pi}\\
        &\leq  \sum_{i \in [r]^n} \Big( 1 + \mathcal{O} \big( \tfrac{4n^2}{d} \big) \Big) \frac{1}{d^{2n}} n! \cdot d^{\bar{n}}  \\
       & \leq  \Big( 1 + \mathcal{O} \big( \tfrac{4n^2}{d} \big) \Big) { \bigg( \frac{r n (d + n)}{d^2} \bigg)}^n.
    \end{align*}

    If $s = t$, then the scalar quantity $(\bra{t}\bra{s})^{\otimes n} \cdot V(\pi) \cdot (\ket{s}\ket{t})^{\otimes n}$ is always equals to $1$ for all permutations $\pi \in \mathfrak{S}_{2n}$. Then a straightforward reordering of the sums yields
    \begin{align}
        \E_U \big[ (X_{ss})^n \big] &\leq \sum_{i \in [r]^n} \Big( 1 + \mathcal{O} \big( \tfrac{4n^2}{d} \big) \Big) \frac{1}{d^{2n}} \sum_{\pi \in {\mathfrak{S}}_{2n}} \prod^{\ell}_{j=1} \Tr \big[ \tilde{K}_{i_{c_j(1)}} \cdots \tilde{K}_{i_{c_j(\ell_j)}} \big] \notag \\
        &\leq \Big( 1 + \mathcal{O} \big( \tfrac{4n^2}{d} \big) \Big) \frac{1}{d^{2n}} \sum_{\pi \in {\mathfrak{S}}_{2n}} \sum_{i \in [r]^n}  \prod^{\ell}_{j=1} \Tr \big[ \tilde{K}_{i_{c_j(1)}} \cdots \tilde{K}_{i_{c_j(\ell_j)}} \big] \label{eq:high_moment_s_eq_t_step}
    \end{align}
    For each permutation $\pi \in \mathfrak{S}_{2n}$ we divide the product $\prod_j \Tr[ \tilde{K}_{i_{c_j(1)}} \cdots \tilde{K}_{i_{c_j(\ell_j)}}]$ into two parts: the cycles of length $1$ (\ie the fixed points of $\pi$) and the larger cycles:
    \begin{equation*}
        \prod^{\ell}_{j=1} \Tr \big[ \tilde{K}_{i_{c_j(1)}} \cdots \tilde{K}_{i_{c_j(\ell_j)}} \big] = \prod^{\ell}_{\substack{j=1 \\ \ell_j = 1}} \Tr \big[ \tilde{K}_{i_{c_j(1)}} \big] \prod^{\ell}_{\substack{j=1 \\ \ell_j \geq 2}}  \Tr \big[ \tilde{K}_{i_{c_j(1)}} \cdots \tilde{K}_{i_{c_j(\ell_j)}} \big]
    \end{equation*}
    then
    \begin{align*}
        \E_U \big[ (X_{ss})^n \big] &\leq \Big( 1 + \mathcal{O} \big( \tfrac{4n^2}{d} \big) \Big) \frac{1}{d^{2n}} \sum_{\pi \in {\mathfrak{S}}_{2n}} \sum_{i \in [r]^n}  \prod^{\ell}_{\substack{j=1 \\ \ell_j = 1}} \Tr \big[ \tilde{K}_{i_{c_j(1)}} \big] \prod^{\ell}_{\substack{j=1 \\ \ell_j \geq 2}}  \Tr \big[ \tilde{K}_{i_{c_j(1)}} \cdots \tilde{K}_{i_{c_j(\ell_j)}} \big] \\
        &\leq \Big( 1 + \mathcal{O} \big( \tfrac{4n^2}{d} \big) \Big) \frac{1}{d^{2n}} \sum_{\pi \in {\mathfrak{S}}_{2n}} \sum_{i \in [r]^n}  \prod^{\ell}_{\substack{j=1 \\ \ell_j = 1}} d^{\#\pi-\mathsf{Fix}(\pi)} \cdot \Tr \big[ \tilde{K}_{i_{c_j(1)}} \big],
    \end{align*}
    where $\mathsf{Fix}(\pi)$ denotes the number of fixed point of $\pi$. Now for each permutation $\pi \in \mathfrak{S}_{2n}$, we divide the sum 
    $\sum_{i}  \prod_{j: \ell_j=1} \Tr[ \tilde{K}_{i_{c_j(1)}}]$ into two parts: the sum over the indices that are fixed by $\pi$, and the complementary part. Let $F_e$ and $F_o$ be the two even and odd fixed points sets defined by
    \begin{align*}
        F_e &\coloneqq \big\{ i \in [n] \;:\; \pi(2i) = 2i \big\} \\
        F_o &\coloneqq \big\{ i \in [n] \;:\; \pi(2i-1) = 2i-1 \big\},
    \end{align*}
    and $F \coloneqq F_e \cup F_o$, such that $|F_e| + |F_o| = \mathrm{Fix}(\pi)$. Then
    \begin{equation*}
        \sum_{i \in [r]^n} \prod^{\ell}_{\substack{j=1 \\ \ell_j = 1}} \Tr \big[ \tilde{K}_{i_{c_j(1)}} \big] = \Bigg( \sum_{i \not\in F} 1 \Bigg) \Bigg( \sum_{i \in F} \prod^{\ell}_{\substack{j=1 \\ \ell_j = 1}} \Tr \big[ \tilde{K}_{i_{c_j(1)}} \big] \Bigg) = r^{n-|F|} \Bigg( \sum_{i \in F} \prod^{\ell}_{\substack{j=1 \\ \ell_j = 1}} \Tr \big[ \tilde{K}_{i_{c_j(1)}} \big] \Bigg),
    \end{equation*}
    and the right-hand-side becomes
    \begin{equation*}
        \sum_{i \in F} \prod^{\ell}_{\substack{j=1 \\ \ell_j = 1}} \Tr \big[ \tilde{K}_{i_{c_j(1)}} \big] = {\bigg( \sum_i \Tr [K_i] \bigg)}^{|F_o|} {\bigg( \sum_i \Tr [K^*_i] \bigg)}^{|F_e|}.
    \end{equation*}
    Using the assumption $\sum_i \vert \Tr(K_i)\vert^2 \leq \phi^2 d^2$ which implies $\sum_i \vert \Tr(K^*_i)\vert = \sum_i \vert \Tr(K_i)\vert \leq \sqrt{r} \phi d$, then
    \begin{equation*}
        \sum_{i \in F} \prod^{\ell}_{\substack{j=1 \\ \ell_j = 1}} \Tr \big[ \tilde{K}_{i_{c_j(1)}} \big] \leq (\sqrt{r}\phi d)^{\mathsf{Fix}(\pi)}.
    \end{equation*}
    Thus,
    \begin{align*}
        \E_U \big[ (X_{ss})^n \big] &\leq \Big( 1 + \mathcal{O} \big( \tfrac{4n^2}{d} \big) \Big) \frac{1}{d^{2n}} \sum_{\pi \in {\mathfrak{S}}_{2n}} \sum_{i \in [r]^n}  \prod^{\ell}_{\substack{j=1 \\ \ell_j = 1}} d^{\#\pi-\mathsf{Fix}(\pi)} \cdot \Tr \big[ \tilde{K}_{i_{c_j(1)}} \big] \\
        &\leq \Big( 1 + \mathcal{O} \big( \tfrac{4n^2}{d} \big) \Big) \frac{1}{d^{2n}} \sum_{\pi \in {\mathfrak{S}}_{2n}} \prod^{\ell}_{\substack{j=1 \\ \ell_j = 1}} \sum_{i \in [r]^n} d^{\#\pi-\mathsf{Fix}(\pi)} \cdot \Tr \big[ K_{i_{c_j(1)}} \big] \\
        &\leq \Big( 1 + \mathcal{O} \big( \tfrac{4n^2}{d} \big) \Big) \frac{1}{d^{2n}} \sum_{\pi \in {\mathfrak{S}}_{2n}} d^{\#\pi-\mathsf{Fix}(\pi)} r^{n - |F|} (\sqrt{r}\phi d)^{\mathsf{Fix}(\pi)} \\
        &\leq \Big( 1 + \mathcal{O} \big( \tfrac{4n^2}{d} \big) \Big) \frac{r^{n}}{d^{2n}} \sum_{\pi \in {\mathfrak{S}}_{2n}} d^{\#\pi} r^{-|F|} {(\sqrt{r}\phi)}^{\mathsf{Fix}(\pi)}
    \end{align*}
    Using that $|F| \geq \mathsf{Fix}(\pi) / 2$, we have
    \begin{align*}
        \E_U \big[ (X_{ss})^n \big] &\leq \Big( 1 + \mathcal{O} \big( \tfrac{4n^2}{d} \big) \Big) \frac{r^{n}}{d^{2n}} \sum_{\pi \in {\mathfrak{S}}_{2n}} d^{\#\pi} r^{-\mathsf{Fix}(\pi)/2} {(\sqrt{r}\phi)}^{\mathsf{Fix}(\pi)} \\
        &\leq \Big( 1 + \mathcal{O} \big( \tfrac{4n^2}{d} \big) \Big) \frac{r^{n}}{d^{2n}} \sum_{\pi \in {\mathfrak{S}}_{2n}} d^{\#\pi} \phi^{\mathsf{Fix}(\pi)} \\
    \end{align*}
    Since $\phi \leq 1$ and using $\mathsf{Fix}(\pi) \geq 2 \#\pi - 2n$ from \Cref{lem:fixed_points_bound}, then
    \begin{align*}
        \E_U \big[ (X_{ss})^n \big] &\leq \Big( 1 + \mathcal{O} \big( \tfrac{4n^2}{d} \big) \Big) \frac{r^n}{d^{2n}} \sum_{c=1}^{2n} \sum_{\substack{\pi \in \mathfrak{S}_{2n} \\ \#\pi = c}} \phi^{2(\#\pi)-2n} d^{\#\pi} \\
        &\leq \Big( 1 + \mathcal{O} \big( \tfrac{4n^2}{d} \big) \Big) \frac{r^n}{d^{2n}} \sum_{c=1}^{2n} {2n \brack c} \phi^{2c-2n} d^{c},
    \end{align*}
    where $\genfrac{[}{]}{0pt}{1}{2n}{c}$ are the unsigned Stirling numbers of the first kind. Using the identity $\sum_{k=0}^n \genfrac{[}{]}{0pt}{1}{n}{k} \cdot x^k = x^{\bar{n}}$, where $x^{\bar{n}}$ denotes the raising factorial $x (x + 1) \cdots (x + n - 1)$, we have
    \begin{equation*}
       \mathbb{E}_U \big[ (X_{ss})^n \big] \leq \Big( 1 + \mathcal{O} \big( \tfrac{4n^2}{d} \big) \Big) \frac{r^{n}}{(\phi d)^{2n}} {\big( \phi^2 d \big)}^{\bar{2n}}.
    \end{equation*}
    Using the upper bound $x^{\bar{n}} \leq (x+n)^n$, we obtain
    \begin{align*}
       \mathbb{E}_U \big[ (X_{ss})^n \big] &\leq \Big( 1 + \mathcal{O} \big( \tfrac{4n^2}{d} \big) \Big) \frac{r^{n}}{(\phi d)^{2n}} {\big( \phi^2 d + 2n \big)}^{2n} \\
        &\leq \Big( 1 + \mathcal{O} \big( \tfrac{4n^2}{d} \big) \Big) {\Big( \sqrt{r} \frac{\phi^{2} d + 2n}{\phi d} \Big)}^{2n}.
    \end{align*}
\end{proof}

Using the upper bounds \Cref{lem:first_moment_upperbound,lem:nth_moment_upperbound} on the moments of the operators $X_{st}$, we can now derive probabilistic bounds on their deviation.

\begin{lemma} \label{lem:deviation_bound}
    Let $\theta$ and $\theta_0$ be two integers such that $\theta>0$ and $0 < \theta_0 \leq \sqrt{d}$. Let $\phi \in [0,1]$ and $\Phi: \mathcal{B}(\CC^d) \to \mathcal{B}(\CC^d)$ a \CPTP map with minimum Kraus rank $r$ and such that $F_e(\Phi) \leq \phi^2$. Then,
    \begin{align*}
        \Pr_U[X_{st} \geq \varepsilon] &\leq e^{-\theta \varepsilon} \mathcal{O} \left( \exp\left( \frac{r \theta \theta_0 (d + \theta_0)}{d^2} \right) + e^{\theta}  \left( \frac{r \theta_0 (d + \theta_0)}{d^2} \right)^{\theta_0} \right), \\[0.5em]
        \Pr_U[X_{ss} \geq \varepsilon] &\leq e^{-\theta \varepsilon} \mathcal{O} \left( \exp\left( \frac{r \theta \left( \phi^2 d + 2\theta_0 \right)^2}{d^2 \phi^2} \right) +e^{\theta} \bigg( \sqrt{r} \frac{\phi^2 d + 2\theta_0}{\phi d} \bigg)^{2\theta_0}  \right).
    \end{align*}
\end{lemma}
\begin{proof}
    In the case $s \neq t$, applying Markov's inequality yields
    \begin{equation*}
        \Pr_U [X_{st} \geq \varepsilon]
        = \Pr_U\left(e^{\theta X_{st}} \geq e^{\theta \varepsilon}\right)
        \leq \frac{\mathbb{E}_U\left[e^{\theta X_{st}}\right]}{e^{\theta \varepsilon}}.
    \end{equation*}
    Therefore,
    \begin{align}
        \Pr_U [X_{st} \geq \varepsilon]
        &\leq  e^{-\theta \varepsilon} \mathbb{E}_U\left[ e^{\theta X_{st}} \right] \notag \\
        &\leq e^{-\theta \varepsilon} \sum_{i=0}^{\infty} \frac{\theta^i}{i!} \mathbb{E}_U[X_{st}^i] \notag \\
        &\leq e^{-\theta \varepsilon} \left[ \sum_{i \leq \theta_0} \frac{\theta^i}{i!} \mathbb{E}_U[X_{st}^i] + \sum_{i \geq \theta_0 + 1} \frac{\theta^i}{i!} \mathbb{E}_U[X_{st}^i] \right] \label{eq:deviation_bound_st_case}
    \end{align}
    We now evaluate these two contributions separately. For the first sum, we have using \Cref{lem:nth_moment_upperbound}
    \begin{align*}
        \sum_{i \leq \theta_0} \frac{\theta^i}{i!} \mathbb{E}_U\left[ X_{st}^i \right] 
        &\leq \sum_{i \leq \theta_0} \frac{\theta^i}{i!} \left[ 1 + \mathcal{O}\left( \frac{4 \theta_0^2}{d} \right) \right] \cdot \frac{r^i (d + \theta_0)^i}{d^{2i}} \\
        &\leq \sum_{i \leq \theta_0} \frac{\theta^i}{i!} \left[ 1 + \mathcal{O}\left( \frac{4 \theta_0^2}{d} \right) \right] \left( \frac{r \theta_0 (d + \theta_0)}{d^2} \right)^i \\
        &\leq \left[ 1 + \mathcal{O}\left( \frac{4 \theta_0^2}{d} \right) \right] \sum_{i \leq \theta_0} \frac{1}{i!} \left( \frac{r \theta \theta_0 (d + \theta_0)}{d^2} \right)^i \\
        &\leq \left[ 1 + \mathcal{O}\left( \frac{4 \theta_0^2}{d} \right) \right] \exp\left( \frac{r \theta \theta_0 (d + \theta_0)}{d^2} \right).
    \end{align*}
    Using the assumption $\theta_0 \leq \sqrt{d}$, we obtain
    \begin{equation} \label{eq:deviation_bound_st_case_low_moments}
        \sum_{i \leq \theta_0} \frac{\theta^i}{i!} \mathbb{E}_U \left[ X_{st}^i \right] \leq \mathcal{O}\left(  \exp\left( \frac{r \theta \theta_0 (d + \theta_0)}{d^2} \right) \right).
    \end{equation}
    For the second sum, we have
    \begin{align*}
        \sum_{i \geq \theta_0 + 1} \frac{\theta^i}{i!} \mathbb{E}_U \left[ X_{st}^i \right] 
        &\leq \sum_{i \geq \theta_0 + 1} \frac{\theta^i}{i!}  \mathbb{E}_U\left[X_{st}^{\theta_0}\right] \\
        &\leq e^{\theta}  \mathbb{E}_U\left[X_{st}^{\theta_0}\right] \\
        & \leq e^{\theta} \left[ 1 + \mathcal{O}\left( \frac{4 \theta_0^2}{d} \right) \right] \left( \frac{r \theta_0 (d + \theta_0)}{d^2} \right)^{\theta_0} \\ 
        &\leq \mathcal{O} \left( e^{\theta}  \left( \frac{r \theta_0 (d + \theta_0)}{d^2} \right)^{\theta_0} \right).
    \end{align*}
    Under the same assumption $\theta_0 \leq \sqrt{d}$, it follows that
    \begin{equation} \label{eq:deviation_bound_st_case_high_moments}
        \sum_{i \geq \theta_0 + 1} \frac{\theta^i}{i!} \mathbb{E}_U\left[X_{st}^i\right] \leq \mathcal{O} \left( e^{\theta}  \left( \frac{r \theta_0 (d + \theta_0)}{d^2} \right)^{\theta_0} \right)
    \end{equation}
    Combining \cref{eq:deviation_bound_st_case,,eq:deviation_bound_st_case_low_moments,,eq:deviation_bound_st_case_high_moments}, we finally obtain
    \begin{equation*}
        \Pr_U [X_{st} \geq \varepsilon] \leq e^{-\theta \varepsilon} \mathcal{O} \left( \exp\left( \frac{r \theta \theta_0 (d + \theta_0)}{d^2} \right) + e^{\theta}  \left( \frac{r \theta_0 (d + \theta_0)}{d^2} \right)^{\theta_0} \right).
    \end{equation*}
        Similarly, for the case $s = t$, we have
    \begin{equation} \label{eq:deviation_bound_ss_case}
        \Pr_U [X_{st} \geq \varepsilon] \leq e^{-\theta \varepsilon} \left[ \sum_{i \leq \theta_0} \frac{\theta^i}{i!} \mathbb{E}_U[X_{ss}^i] + \sum_{i \geq \theta_0 + 1} \frac{\theta^i}{i!} \mathbb{E}_U[X_{ss}^i] \right].
    \end{equation}
    For the first sum, corresponding to low-order moments, we obtain
    \begin{align}
        \sum_{i \leq \theta_0} \frac{\theta^i}{i!} \mathbb{E}_U[X_{ss}^i]
        &\leq \sum_{i \leq \theta_0} \frac{\theta^i}{i!} \left(1 + \frac{4 i^2}{d} \right) \left( \sqrt{r} \frac{\phi^2 d + 2i }{\phi d} \right)^{2i} \notag \\
        &\leq \sum_{i \leq \theta_0} \frac{\theta^i}{i!} \left(1 + \frac{4\theta_0^2}{d} \right) \left( \sqrt{r} \frac{\phi^2 d + 2\theta_0}{\phi d} \right)^{2i} \notag \\
         &\leq \mathcal{O} \left( \sum_{i \leq \theta_0} \frac{\theta^i}{i!} \left( \sqrt{r} \frac{\phi^2 d + 2\theta_0}{\phi d} \right)^{2i} \right) \notag \\
        &\leq \mathcal{O} \left( \sum_{i \leq \theta_0} \frac{1}{i!} \left[ \frac{r \theta \left( \phi^2 d + 2\theta_0 \right)^2}{(\phi d)^2} \right]^i\right) \notag \\
        & \leq \mathcal{O} \left( \exp\left( \frac{r \theta \left( \phi^2 d + 2\theta_0 \right)^2}{d^2 \phi^2} \right) \right). \label{eq:deviation_bound_ss_case_low_moments}
    \end{align}
    For the second sum, corresponding to high-order moments, we have
    \begin{align}
        \sum_{i \geq \theta_0 + 1} \frac{\theta^i}{i!} \mathbb{E}_U\left[X_{ss}^i\right] 
        &\leq \sum_{i \geq \theta_0 + 1} \frac{\theta^i}{i!}  \mathbb{E}_U\left[X_{ss}^{\theta_0}\right] \notag \\
        &\leq e^{\theta}  \mathbb{E}_U\left[X_{ss}^{\theta_0}\right] \notag \\
        & \leq e^{\theta} \Big( 1 + \mathcal{O} \big( \tfrac{4\theta_0 ^2}{d} \big) \Big)  {\bigg( \sqrt{r} \frac{\phi^2 d + 2\theta_0}{\phi d} \bigg)}^{2\theta_0}  \notag \\
        & \leq \mathcal{O}\left( e^{\theta} \bigg( \sqrt{r} \frac{\phi^2 d + 2\theta_0}{\phi d} \bigg)^{2\theta_0} \right). \label{eq:deviation_bound_ss_case_high_moments}
    \end{align}
    Combining the above estimates \cref{eq:deviation_bound_ss_case,,eq:deviation_bound_ss_case_low_moments,,eq:deviation_bound_ss_case_high_moments}, we conclude that
    \begin{equation*}
        \Pr_U [X_{ss} \geq \varepsilon] \leq e^{-\theta \varepsilon} \mathcal{O} \left( \exp\left( \frac{r \theta \left( \phi^2 d + 2\theta_0 \right)^2}{d^2 \phi^2} \right) +e^{\theta} \bigg( \sqrt{r} \frac{\phi^2 d + 2\theta_0}{\phi d} \bigg)^{2\theta_0}  \right).
    \end{equation*}
\end{proof}

\subsection{Proof of \Cref*{thm:tamper_detection_code}} \label{sec:tamper_detection_code_proof}

We are now in a position to establish the main result stated in \Cref{thm:tamper_detection_code}.
\begin{proof}[{Proof of \Cref{thm:tamper_detection_code}}]
    We are going to prove that for all messages $s \in \mathcal{M}$, and for all \CPTP map $\Phi \in \advCPTP$, the following holds:
    \begin{equation*}
        \Pr \left[ \dec_U \circ {\Phi} \circ \enc_U(s) = \perp \right] \geq 1 - \varepsilon,
    \end{equation*}
    where the probability is over the randomness of the measurement of $\dec_U$. Equivalently, that for all messages $s \in \mathcal{M}$, and for all \CPTP map $\Phi \in \advCPTP$, the probability that the decoder fails to detect the tampering is at most $\varepsilon$. Thus using \Cref{lem:deviation_bound} and
    \begin{equation*}
        \Pr \left[ \dec_U \circ {\Phi} \circ \enc_U(s) = \perp \right] = 1 - \sum_{t \in \mathcal{M}} X_{st},
    \end{equation*}
    we have
    \begin{align*}
        \Pr_U \Big[ \Pr \big[ \dec_U \circ \Phi \circ \enc_U(s) \neq \perp \big] \geq \varepsilon \Big]
        &\leq \Pr_U \bigg[ \exists t: X_{st} \geq \frac{\varepsilon}{2^k} \bigg] \\
        &\leq \Pr_U \bigg[ X_{ss} \geq \frac{\varepsilon}{2^k} \bigg] + \Pr_U \bigg[ \exists t \neq s \ :\  X_{st} \geq \frac{\varepsilon}{2^k} \bigg] \\
        &\leq  e^{-\theta \varepsilon/ 2^k} \mathcal{O} \Bigg( \exp\left( \frac{r \theta \left( \phi^2 d + 2\theta_0 \right)^2}{d^2 \phi^2} \right) +e^{\theta} \bigg( \sqrt{r} \frac{\phi^2 d + 2\theta_0}{\phi d} \bigg)^{2\theta_0} \\
        &\quad+ 2^k \exp\left( \frac{\theta r \theta_0 (d + \theta_0)}{d^2} \right) + 2^k e^{\theta}  \left( \frac{r \theta_0 (d + \theta_0)}{d^2} \right)^{\theta_0} \Bigg),
    \end{align*}
    where the inner probabilities are taken over the inherent quantum random dues to quantum decoders (measurements). Taking union bound over $s \in \mathcal{M}$ and $\Phi \in \advCPTP$, with $|\mathcal{M}| = 2^k$, and using  \Cref{lem:deviation_bound}, 
    \begin{align*}
         &\Pr_U \Big[ \exists s, \Phi \ : \ \Pr \big[ \dec_U \circ \Phi \circ \enc_U(s) \neq \perp \big] \geq \varepsilon \Big] \\
         &\leq 2^{k} \vert \advCPTP \vert e^{-\theta \varepsilon / 2^k} \mathcal{O} \Bigg( \exp\left( \frac{r \theta \left( \phi^2 d + 2\theta_0 \right)^2}{d^2 \phi^2} \right) +e^{\theta} \bigg( \sqrt{r} \frac{\phi^2 d + 2\theta_0}{\phi d} \bigg)^{2\theta_0} \\
         &\quad+ 2^k \exp\left( \frac{\theta r \theta_0 (d + \theta_0)}{d^2} \right) + 2^k e^{\theta}  \left( \frac{r \theta_0 (d + \theta_0)}{d^2} \right)^{\theta_0} \Bigg) \\
         &\leq 2^{k} \vert \advCPTP \vert e^{-\theta \varepsilon / 2^k} \mathcal{O} \Bigg( \exp\left( \frac{r \theta \left( \phi^2 d + 2\theta_0 \right)^2}{d^2 \phi^2} \right) +e^{\theta} \bigg( \sqrt{r} \frac{\phi^2 d + 2\theta_0}{\phi d} \bigg)^{2\theta_0} \\
        &\quad+ 2^k \exp\left( \frac{2\theta r \theta_0}{d} \right) + 2^k e^{\theta}  \left( \frac{2r \theta_0}{d} \right)^{\theta_0} \Bigg).
    \end{align*}
    Setting $r \leq d^{1-\delta}$ and $\phi = \sqrt{\frac{2 \theta_0}{d}}$, with $\delta \in [0,1)$, we obtain
    \begin{align*}
        &\Pr_U \Big[ \exists s, \Phi \ : \ \Pr\big[ \dec_U \circ \Phi \circ \enc_U (s) \neq \perp \big] \geq \varepsilon \Big] \\
        &\leq 2^{k} \vert \advCPTP \vert e^{-\theta \varepsilon / 2^k} \mathcal{O} \Bigg( \exp\left( \frac{8 \theta \theta_0}{d^\delta} \right) +e^{\theta} \bigg( \frac{8 \theta_0}{d^\delta} \bigg)^{\theta_0} + 2^k \exp\left( \frac{2\theta \theta_0}{d^{\delta}} \right) + 2^k e^{\theta}  \left( \frac{2\theta_0}{d^{\delta}} \right)^{\theta_0} \Bigg).
    \end{align*}
    Setting, $\theta = d^{\beta}$  and $\theta_0 = d^{\delta/2}$, where $\beta$ is a constant such that $\alpha < \beta \leq \delta/2$, we have
    \begin{align*}
         &\Pr_U \Big[ \exists s, \Phi \ : \ \Pr\big[ \dec_U \circ \Phi \circ \enc_U (s) \neq \perp \big] \geq \varepsilon \Big] \\
        &\leq 2^{k} \vert \advCPTP \vert \exp\left(-\frac{d^\beta \varepsilon }{ 2^k}\right) \mathcal{O} \Bigg( \exp\left( \frac{8 d^{\beta} }{d^{\delta/2}} \right) + e^{d^{\beta}} \bigg(\frac{8}{d^{\delta/2}} \bigg)^{d^{\delta/2}}  + 2^k \exp\left( \frac{2 d^{\beta} }{d^{\delta/2}} \right) + 2^k e^{d^{\beta}} \bigg( \frac{2}{d^{\delta/2}} \bigg)^{d^{\delta/2} }\Bigg) \\
        &\leq 2^{k} \vert \advCPTP \vert \exp\left(-\frac{d^\beta \varepsilon }{ 2^k}\right) \mathcal{O} \Bigg(  2^k \exp\left( \frac{8 d^{\beta} }{d^{\delta/2}} \right) + 2^k e^{d^{\beta}} \bigg(\frac{8}{d^{\delta/2}} \bigg)^{d^{\delta/2}} \Bigg) \\
        &\leq 2^{2k} \vert \advCPTP \vert \exp\left(-\frac{d^\beta \varepsilon }{ 2^k}\right) \mathcal{O} \Bigg(  \exp\left(  {8 d^{\beta - \delta/2 } } \right) +   e^{d^{\beta}} \bigg(\frac{8}{d^{\delta/2}} \bigg)^{d^{\delta/2}} \Bigg). 
    \end{align*}
    For $\beta \leq \frac{\delta}{2}$ and large enough $d$, it is not hard to see that $\exp\left(  {8 d^{\beta - \delta/2 } } \right)  \geq   e^{d^{\beta}} \bigg(\frac{8}{d^{\delta/2}} \bigg)^{d^{\delta/2}}$. Thus,
    \begin{equation*}
        \Pr_U \Big[ \exists s, \Phi \ : \ \Pr\big[ \dec_U \circ \Phi \circ \enc_U (s) \neq \perp \big] \geq \varepsilon \Big] \leq 2^{2k} \vert \advCPTP \vert \exp\left(-\frac{d^\beta \varepsilon }{ 2^k}\right) \mathcal{O} \Bigg(  \exp\left(  {8 d^{\beta - \delta/2 } } \right) \Bigg).
    \end{equation*}
    Using that $\vert \advCPTP \vert = \exp\left(d^{\alpha}\right)$, we finally have
    \begin{align*}
        \Pr_U \Big[ \exists s, \Phi \ : \ \Pr\big[ \dec_U \circ \Phi \circ \enc_U (s) \neq \perp \big] \geq \varepsilon \Big] 
        &\leq 2^{2k}  \mathcal{O} \Bigg(  \exp\left(  {8 d^{\beta - \delta/2 } -\frac{d^\beta \varepsilon }{ 2^k} } + d^{\alpha} \right) \Bigg) \\
        &\leq \mathcal{O} \Bigg(  \exp\left(  {8 d^{\beta - \delta/2 } -\frac{d^\beta \varepsilon }{ 2^k} } + d^{\alpha} +2k \right) \Bigg).
    \end{align*}
    Recall that the expansion factor $\gamma$ is equals to $\frac{\log d}{\log \mathcal{M}}$, such that $\log(d) = \gamma k$, then when $\beta - \alpha >  2/ \gamma $ and $\varepsilon = 2^{\shortminus k}$, the right-hand side term becomes exponentially small in $k$.
\end{proof}
\subsection{The classical adversary}

Recall that the classical theory of tamper detection considers adversarial functions
$\,f:\{0,1\}^n \to \{0,1\}^n\,$.
We can lift this naturally to the quantum setting via CPTP extensions.

\begin{definition}[CPTP extension of a classical function]
Let $\{\lvert x \rangle\}_{x\in\{0,1\}^n}$ denote the computational basis of $(\mathbb{C}^2)^{\otimes n}$.
A completely positive, trace-preserving (CPTP) map
$\Phi:\mathcal{B}((\mathbb{C}^2)^{\otimes n})\to\mathcal{B}((\mathbb{C}^2)^{\otimes n})$
\emph{extends} a classical function $f$ if for every $x\in\{0,1\}^n$,
\begin{equation}
  \Phi\bigl(\lvert x\rangle\!\langle x\rvert\bigr) \,=\, \lvert f(x)\rangle\!\langle f(x)\rvert.
\end{equation}
We call such a $\Phi$ \emph{classical}.
Equivalently, $\Phi$ maps computational-basis pure states to computational-basis pure states according to $f$.
\end{definition}

For an adversary holding a family of channels $\mathcal{F}_{\mathrm{adv}}$, we say the adversary \emph{acts classically} (or \emph{acts via classical extensions}) if every $\Phi\in\mathcal{F}_{\mathrm{adv}}$ is the CPTP extension of some classical function $f_{\Phi}$.

\begin{lemma} \label{lem:properties_of_classical_maps}
    Let $f$ be a classical function from $n$-bits to $n$-bits. 
    Consider its {\CPTP extension} $\Phi$.
    Then,
    \begin{itemize}
        \item $\rank(\Phi) =  \max_{y \in \{0,1\}^n} |f^{-1}(y)| . $
        \item $F_e(\Phi) \in \left[ \frac{1}{2^{2n} r} \text{(number of fixed points of $f$)}^2,   \frac{1}{2^{2n}} \text{(number of fixed points of $f$)}^2 \right]$.
    \end{itemize}
\end{lemma}

\begin{proof}
    
We divide the proof into two parts.
First we compute the rank and then we move onto bounds for the entanglement fidelity.
The proof for the rank itself is further divided into two parts.
We show that any Stinespring representation requires dimension of the environment to be at least $|f^{-1}(y)|$ (for each $y \in \lbrace 0,1 \rbrace^n$).
This, using correspondence between Stinespring and Kraus representations, gives a lower bound on $\rank(\Phi)$.
Then, we give an explicit Kraus representation, which naturally give an upper bound on $\rank(\Phi)$.
This explicit Kraus representation also helps us to compute bounds on the entanglement fidelity (which as described before, is invariant under the choice of Kraus decomposition.) 

\textsf{Lower Bound:}

Let $f:\{0,1\}^n\to\{0,1\}^n$ and $d=2^n$. Consider any CPTP map $\Phi$ that
satisfies
\[
\Phi(|x\rangle\langle x|)=|f(x)\rangle\langle f(x)|\qquad(\forall x).
\]
Consider a Stinespring isometry $V:\mathcal{H}\to\mathcal{H}\otimes\mathcal{H}_E$ such that
$\Phi(\rho)=\operatorname{Tr}_E[V\rho V^\dagger]$. Then for each input basis
element there exist environment vectors $ \ket{\psi_x}_E$ such that
\[
V\ket{x}=\ket{f(x)}\otimes \ket{\psi_x}_E.
\]
Since $V$ is an isometry, it preserves inner product and thus, we have
\[
\langle x|x'\rangle=\langle Vx|Vx'\rangle=\langle f(x)|f(x')\rangle\,
\langle\psi_x|\psi_{x'}\rangle.
\]
In other words, for $x \neq x^\prime$ such that $f(x) =f(x^\prime)$, $\ket{\psi_x}$ and $\ket{\psi_{x^\prime}}$ must be orthogonal.
So the environment must contain orthonormal vectors for each of those $x$'s with the same output.
Hence,
\[ \rank(J (\Phi)) \geq 
\dim\mathcal{H}_E \;\ge\; \max_{y\in\mathrm{Im}(f)} |\preim(y)|.
\]

\textsf{Upper bound:} 

Let $f : \{0,1\}^n \to \{0,1\}^n$.  
A collection $\{B_1, B_2, \dots, B_\ell\}$ is called a \emph{collision-free partition} of $\{0,1\}^n$ for $f$ if  
\begin{itemize}
\item $\{B_1, B_2, \dots, B_\ell\}$ is a partition.
That is, $\bigcup_{i=1}^\ell B_i = \{0,1\}^n, \quad B_i \cap B_j = \varnothing \ \text{for all } i \neq j$.
\item  There is no collision for any block.
That is, for any $i \in [l]$ and for any $x,y \in B_i,\ x \neq y$ we have  $f(x) \neq f(y).$
Equivalently, $\left. f \right|_{B_i}$ is injective for all $i \in [\ell]$.
\end{itemize}

It is not hard to see that one can achieve a collision-free partition for $l = \max\limits_{y\in \mathrm{Im}(f)} { \vert f^{-1}(y)\vert}$. 
Simply, for each \(y \in \{0,1\}^n\)  enumerate
\[
f^{-1}(y) = \{ x_{y,1}, x_{y,2}, \dots, x_{y,|f^{-1}(y)|} \}.
\]
Let \(l := \max_{y \in \{0,1\}^n} |f^{-1}(y)|\), and define
\[
B_k := \{\, x_{y,k} \mid y \in \{0,1\}^n, \ k \le |f^{-1}(y)| \}, 
\quad k = 1, \dots, l.
\]
Then \(\{B_1, \dots, B_m\}\) is a collision-free partition: 
the \(B_k\) are pairwise disjoint, their union is \(\{0,1\}^n\), 
and each \(B_k\) contains at most one element from any fibre \(f^{-1}(y)\), 
so \(f\) is injective on each \(B_k\).

Now, we can consider the Kraus operators $K_k = \sum\limits_{x \in B_k} \ketbra{f(x)}{x}$.
Again, it is easy to see that these form Kraus operators for $f$.
Thus,
\[ \rank(J(\Phi)) \leq \mathrm{Krauss\ rank}  \leq  \max_{y \in \{0,1\}^n} |f^{-1}(y)| .\]

Combining lower bound and upper bound, we get 
\[ \rank(J(\Phi)) =  \max_{y \in \{0,1\}^n} |f^{-1}(y)| .\]

Moreover,

\begin{align*}
    \sum_k \vert \Tr(K_k)\vert^2 & = \sum_k \left( \Tr \sum\limits_{x \in B_k} \ketbra{f(x)}{x} \right)^2 \\
     & = \sum_k \left(\sum_{x^\prime} \langle x^\prime \sum\limits_{x \in B_k} \ketbra{f(x)}{x} x^\prime \rangle \right)^2 \\
      & = \sum_k \left(   \sum\limits_{x \in B_k} \braket{f(x)}{x}   \right)^2 \\
      & \leq \left( \sum_k    \sum\limits_{x \in B_k} \braket{f(x)}{x}   \right)^2 \\
      & \leq \left( \sum_{x \in \{0,1\}^n}  \braket{f(x)}{x} \right)^2 \\
      & \leq \text{(number of fixed points of $f$)}^2.
\end{align*}

\begin{align*}
    \sum_k \vert \Tr(K_k)\vert^2 & = \sum_k \left( \Tr \sum\limits_{x \in B_k} \ketbra{f(x)}{x} \right)^2 \\
     & = \sum_k \left(\sum_{x^\prime} \langle x^\prime \sum\limits_{x \in B_k} \ketbra{f(x)}{x} x^\prime \rangle \right)^2 \\
      & =  \sum_k \left(   \sum\limits_{x \in B_k} \braket{f(x)}{x}   \right)^2 \\
      &  \geq \frac{1}{r} \left( \sum_k    \sum\limits_{x \in B_k} \braket{f(x)}{x}   \right)^2 \\
      & \geq \frac{1}{r}\left( \sum_{x \in \{0,1\}^n}  \braket{f(x)}{x} \right)^2 \\
      & \geq \frac{1}{r} \text{(number of fixed points of $f$)}^2.
\end{align*}

Thus, 

\[  \frac{1}{r} \text{(number of fixed points of $f$)}^2 \leq  \sum_k \vert \Tr(K_k)\vert^2 \leq  \text{(number of fixed points of $f$)}^2. \]
Dividing all the three terms by $d^2=(2^n)^2$ gives us the bounds on the entanglement fidelity.
\end{proof}

\begin{corollary}[Classical min-entropy form]\label{cor:classical_min_entropy}
Let $n,k\in\mathbb{N}$ and set $d=2^n$. Let $\mathcal F$ be a family of  functions
\(
f:\{0,1\}^n\to\{0,1\}^n.
\)
Define
\[
p_{\max}(f):=\max_{y}\Pr[f(X)=y],\qquad
p_{\mathrm{fix}}(f):=\Pr[f(X)=X].
\]
Fix parameters $0\le 2\alpha<\delta<1$. Suppose every $f\in\mathcal F$ satisfies
\begin{enumerate}
    \item (min-entropy / rank) \quad $H_\infty(f(X)) \ge n\delta$, i.e.
    \[
    p_{\max}(f)\le 2^{-n\delta},
    \]
    \item (fixed-point probability) \quad
    \[
    p_{\mathrm{fix}}(f)\le \sqrt{2}\,2^{-n\delta/4},
    \]
\end{enumerate}
and the family size obeys
\[
|\mathcal F|\le 2^{d^{\alpha}} = 2^{2^{\alpha n}}.
\]
Then for a Haar-random unitary $U\in\mathcal U_d(\mathbb C)$,
\[
\Pr_U\!\big[ (\dec_U,\enc_U)\ \text{is }\varepsilon\text{-tamper secure against }\mathcal F \big]
\;\ge\; 1-\operatorname{negl}(k),
\]
with $\varepsilon=\operatorname{negl}(k)$ (and $\gamma=\mathcal O(1)$). In particular there exists a fixed unitary $U$ achieving this security.
\end{corollary}

\begin{proof}[Proof sketch]
Let $r_f:=\max_y|f^{-1}(y)|$. Since $p_{\max}(f)=r_f/2^n$, condition (i) is equivalent to
$r_f\le d^{1-\delta}$, so the rank condition of Theorem \ref{thm:tamper_detection_code} holds.
Using  Lemma~\ref{lem:properties_of_classical_maps}, we have $F_e(\Phi_f)\le p_{\mathrm{fix}}(f)^2$, and then condition (ii)
implies $F_e(\Phi_f)\le 2\,d^{-\delta/2}$, which is the entanglement-fidelity bound required by the theorem.
Finally, the cardinality bound on $\mathcal F$ is the family-size condition, which is readily satisfied. Applying Theorem
\ref{thm:tamper_detection_code} yields the stated conclusion.
\end{proof}

\subsection{Unitary tamper detection}
Now we show how our theorem implies the unitary tamper detection as given by~\cite{BK23}.

\begin{corollary}\label{cor:unitary_tamperers_delta1}
Let $\advU \subseteq \mathcal{U}_d(\CC)$ be an adversarial family of unitary operators where each $V\in\advU$ induces the channel $\rho\mapsto V\rho V^\dagger$.  
Assume the family satisfies
\begin{enumerate}
    \item (family size) \quad $|\advU|\le 2^{d^{\alpha}}$ for some $\alpha<\tfrac12$,
    \item (trace bound) \quad for every $V\in\advU$,
    \[
      \big|\Tr(V)\big|\le \sqrt{2}\,d^{3/4}.
    \]
\end{enumerate}
Then, for a Haar-random unitary $U\in\mathcal U_d(\CC)$,
\[
  \Pr_U\!\big[ (\dec_U,\enc_U)\ \text{is }\varepsilon\text{-tamper secure against }\advU\big]
  \;\ge\; 1-\operatorname{negl}(k),
\]
with $\varepsilon=\operatorname{negl}(k)$ and $\gamma=\mathcal O(1)$. In particular, for any fixed $\alpha<\tfrac12$ and sufficiently large $d$ there exists a (fixed) unitary $U$ such that $(\dec_U,\enc_U)$ is $\varepsilon$-tamper secure against every $V\in\advU$ satisfying the above trace and size constraints.
\end{corollary}

\begin{proof}

First, note that since all the operators in $\advU$ are unitary, $\rank{\advU}=1$.
Hence, for any constant $\delta \geq 0$, the rank condition required by our theorem is satisfied.
Now since $\alpha<\tfrac{1}{2}$ we may pick any $\delta <1$ such that $2\alpha<\delta<1$.
This further implies that for a sufficiently large $d$,
$\frac{2}{d^{1/2}}\le \frac{2}{d^{\delta/2}}$.
Thus, it suffices to show that $F_e(\advU) \leq \frac{2}{d^{1/2}}$, which will complete the third and final requirement of the theorem, giving us the required corollary.
Now, compute the entanglement fidelity of a unitary channel induced by $V$: with $|\Omega\rangle=\tfrac{1}{\sqrt d}\sum_{i=1}^d|ii\rangle$.

\[ F_e(V)=\big|\langle\Omega|(V\otimes I)|\Omega\rangle\big|^2=\frac{|\operatorname{Tr}(V)|^2}{d^2}.
\]

By the assumption of the corollary, $|\operatorname{Tr}(V)|\le \sqrt{2} d^{3/4}$.
So,
\[F_e(V)\le \frac{2 d^{3/2}}{d^2}=\frac{2}{d^{1/2}}.
\]
\qedhere
\end{proof}

\section{Towards Universal Tamper Detection} \label{sec:universal_tamper}

In this section, we provide the definitions and context for our conjecture on universal quantum tamper detection (\Cref{sec:universal-conjecture}), and then, in support of the conjecture, we show a quantum advantage of the Haar random scheme for tamper detection (\Cref{sec:advantage-tamper}). 

\subsection{Universal quantum tamper detection}
\label{sec:universal-conjecture}
Recall that \cite{JW15} showed that there are two types of bad families for classical tamper detection. 
These families are constant functions and identity functions.
There are two reasons one can chose to ``ignore" the identity function (and those that are close to it).
First, it is impossible to protect against them, irrespective form of encoding-decoding scheme consider, simply from the definition. 
There is no way to catch an adversary who does nothing (applies the identity map). 
The second reason is that such functions do not cause any harm to the message.
The decoder does not detect tampering, but it still produces the correct decoding.
So, one can only focus only on  non-trivial maps, considering a slightly weaker notion of security, but which still captures the spirit of tamper detection.
\cite{BK23} refer to this as \emph{relaxed tamper detection}.

\begin{definition}[Relaxed tamper detection\label{def:relaxed_TD}]
Let $\mathcal{F}_{\mathsf{Adv}}$ be a family of functions from $\{0,1\}^n$ to $\{0,1\}^n$ and let $\mathcal{M}= \lbrace 0,1\rbrace^k$ be the message set. 
We say that an encoding-decoding scheme $\left(\enc, \dec\right)$ is $\varepsilon$-tamper secure in the relaxed setting (against  $\mathcal{F}_{\mathsf{Adv}}$), if for all $m \in \mathcal{M}$ and for all $f \in \mathcal{F}_{\mathsf{Adv}}$, the following holds:
\[\Pr\Big[ \dec   \circ f \circ \enc(m) \in \lbrace \perp,m \rbrace \Big] \geq 1-\varepsilon.\] \end{definition}

Although \cite{JW15} did not explicitly define relaxed version of tamper detection, it follows from their observations that there is a bad family of size only $2^n$ for relaxed tamper detection.
We briefly sketch the proof partly for completeness and partly to guide our discussion for the quantum case. 

\begin{lemma}[\cite{JW15}] \label{lem:Const_against_classical_rtd}
    There exists a family $\advcls$ of size $2^n$ such that there does not exist a relaxed tamper detection code (with a classical $\enc-\dec$) against $\advcls$.
\end{lemma}
\begin{proof}
$\mathcal F_{\mathrm{const}}=\{f_y:x\mapsto y\mid y\in\{0,1\}^n\}$ be the family of all constant functions on $n$ bits.
Assume for contradiction that $(\enc,\dec)$ is an $\varepsilon$-relaxed tamper-detection code against $\mathcal F_{\mathrm{const}}$.
The set of possible $n$-bit strings has size $2^n$, hence there exists some $c\in\{0,1\}^n$ and a message $m_0\in\mathcal M$ with $\Pr[\enc(m_0)=c]>0$.
Consider the constant tampering function $f_c\in\mathcal F_{\mathrm{const}}$.
For any $m_1\in\mathcal M$ with $m_1\neq m_0$ we have
\[
\dec(f_c(\enc(m_1)))=\dec(c),
\]
which equals $m_0$ with positive probability (when $\enc(m_0)=c$).
Thus for such $m_1$ the decoder outputs a wrong message (neither $\perp$ nor $m_1$) with non-negligible probability,
violating the relaxed tamper-detection requirement.
Therefore no relaxed tamper-detection scheme against $\mathcal F_{\mathrm{const}}$ exists.
\end{proof}

The above lemma says that there exists no classical relaxed tamper detection scheme that prevents against $\mathcal F_{\mathrm{const}}$.  However, it is not hard to see that one can use unitary schemes to protect against $\mathcal F_{\mathrm{const}}$.

\begin{proposition} \label{prop:unitary_against_constant}
Let $\mathcal F_{\mathrm{const}}=\{f_y:x\mapsto y\mid y\in\{0,1\}^n\}$ be the family of all classical constant functions on $n$ bits.
There exists a unitary encoding--decoding pair $(\enc,\dec)$ acting on $n$ qubits such that, for every $m\in\mathcal M$ and every $y\in\{0,1\}^n$,
\[
\Pr\big[\dec\circ f_y\circ\enc(m)\in\{\perp,m\}\big]\ge 1-\operatorname{negl}(k),
\]
provided $n\ge 3k$. In other words, the scheme is $\varepsilon$-relaxed-tamper-secure against $\mathcal F_{\mathrm{const}}$ with $\varepsilon=\operatorname{negl}(k)$.
\end{proposition}

\begin{proof}
    Model a classical constant tampering $f_y$ as the quantum channel that replaces any $n$-qubit state by the computational-basis state $\ket{y}$ (i.e.\ the CPTP map $\rho\mapsto\ket{y}\!\bra{y}$).
    Let the encoder be the unitary that maps a basis message $\ket{m}\ket{0^{\,n-k}}$ to the Hadamard-coded state
    \[ 
    \enc(m) = \ket{\psi_m}:=H^{\otimes n}\big(\ket{m}\ket{0^{\,n-k}}\big).
    \]
    The decoder, as usual, first inverts the encoder and then measures in the computational basis.
    \begin{itemize}
      \item outputs the message $m'$ if the measured string equals $(m',0^{\,n-k})$ (i.e. the last $n-k$ bits are all zero)
      \item otherwise outputs $\perp$ (abort).
    \end{itemize}
    Now the probability that decoder outputs $m' \in \lbrace 0,1 \rbrace^k$ is
    \begin{align*}
        \Pr\big[\dec\circ f_y\circ\enc(m)= m'\big]  & = \Big\vert \bra{m'}\bra{0^{n-k}}  H^{\otimes n}\ket{y}  \Big\vert^2 \\
        & \leq \Big\vert \bra{0^{n-k}}  H^{\otimes {n-k}} \ket{y^{n}_{n-k+1}}  \Big\vert^2
    \end{align*}
    where ${y^{n}_{n-k+1}}$ are the last $n-k$ bits of $y$.
    One can check that for any string $y \in \{ 0,1 \}^n$, the above quantity is $2^{k-n}$.
    Thus, 
    \begin{align*}
        \Pr\big[\dec\circ f_y\circ\enc(m)\in\{\perp,m\}\big] &= 1 -  \Pr\big[\dec\circ f_y\circ\enc(m)= m' \text{ for some $m' \not\in \{m, \perp\}$ }\big] \\
        & \geq \sum_{m \neq m', \perp} \Pr\big[\dec\circ f_y\circ\enc(m)= m' \big] \\
        &\geq 1- \sum_{m \neq m', \perp} 2^{k-n} \\
        & \geq 1-  2^{2k-n},
    \end{align*}
    as claimed.
\end{proof}

This motivates the following conjecture.
\begin{conjecture} [Universal Quantum Tamper Detection]\label{con:tamper_dection_against_any_poly}
Let $\mathcal{M} = \lbrace 0,1 \rbrace^k$ and 
    let $\advcls$ be any family of CPTP operators of size at most $2^{2^{\alpha n}}$, for some constant $\alpha <\frac{1}{2}$. 
    Then, there exists an $\varepsilon$-relaxed tamper detection code against ${\advcls}$ with $\varepsilon = \operatorname{negl}(k)$ and $\gamma = \mathcal{O}(1)$. 
\end{conjecture}

Note that both \Cref{lem:Const_against_classical_rtd} as well as \Cref{prop:unitary_against_constant} hold even if one replaces relaxed tamper detection by tamper detection. 
However, we chose to present in the relax tamper detection form as \Cref{con:tamper_dection_against_any_poly} is provably false if we replace relaxed tamper detection by tamper detection.

\subsection{Quantum advantage for tamper detection}
\label{sec:advantage-tamper}

In contrast to the classical case, we are able to establish the following result, which constitutes a special case of 
\Cref{con:tamper_dection_against_any_poly}.
For this section, we fix $\mathcal{M} = \lbrace 0,1 \rbrace^k$.

\begin{theorem}[Tamper detection against replacement channels]\label{thm:tamper_dection_against_any_constants}
    Let $\mathrm{States} = \lbrace \ket{\psi}\rbrace_\psi$ be  a finite collection of states.
    Let $\Phi_{\psi}$ be a \CPTP map such that $\Phi_{\psi}(\rho)= \ketbra{\psi}{\psi}$, for any $\rho$. 
    Let $\advcls(\mathrm{States}) = \lbrace \Phi_{\psi} :\ket{\psi} \in C\rbrace$.
    If $\vert \mathrm{States} \vert \leq 2^{2^{\alpha n}} =  2^{d^{\alpha}}$ (for some $\alpha <1$), then there exists a tamper detection code against $\advcls(\mathrm{States})$ with $\varepsilon = \operatorname{negl}(k)$ and $\gamma = \mathcal{O}(1)$.
\end{theorem}

We would like to remind the reader that classically, for relaxed tamper detection, the only bad family was that of constant functions. 
Replacement channels are generalizations of constant functions.
Theorem~\ref{thm:tamper_dection_against_any_constants} shows that such families are not bad when one considers quantum encodings.
Before going to the proof of Theorem~\ref{thm:tamper_dection_against_any_constants} we will need the following lemma.

\begin{lemma} \label{lem:unit_sphere_scalar_probabilty}
    Let $v$ and $w$ be two vectors in $\CC^d$. Then we have the following: for all $0 \leq t \leq \norm{v}^2 \norm{w}^2$,
    \begin{equation*}
        \Pr_U \Big[ |\langle v, U w \rangle|^2 \geq t \Big] = {\left( 1 - \frac{t}{\norm{v}^2 \norm{w}^2} \right)}^{d-1},
    \end{equation*}
    where the probability is over Haar random unitary $U \in \mathcal{U}_d(\CC)$.
\end{lemma}
\begin{proof}
    If $\norm{v} = 0$ or $\norm{w} = 0$, the claim is immediate. Assume instead that $\norm{v}, \norm{w} > 0$. By Haar invariance, the vector 
    \begin{equation*}
        g \coloneqq \frac{U w}{\norm{w}},
    \end{equation*}
    is uniformly distributed over the unit sphere in $\CC^d$, and hence the distribution of $|\langle v, U w \rangle|$ depends only on $\norm{v}$ and $\norm{w}$. Defining
    \begin{equation*}
        u \coloneqq \frac{v}{\norm{v}},
    \end{equation*}
    we obtain
    \begin{equation*}
        |\langle v, U w \rangle|^2 = \norm{v}^2 \norm{w}^2 \cdot |\langle u, g \rangle|^2.
    \end{equation*}
    By unitary symmetry, the random variable $|\langle u, g \rangle|^2$ has the same distribution as $|g_1|^2$, where $g_1$ is the first coordinate of $g$ (take $u = e_1$, the first standard basis vector). Let $z = (z_1, \ldots, z_d)$ be a vector of i.i.d. standard complex normal variables, i.e., $z_i \sim \mathcal{CN}(0,1)$. Define
    \begin{equation*}
        X \coloneqq |z_1|^2.
    \end{equation*}
    Writing $z = x + i y$ with $x,y \sim \mathcal{N}(0,1/2)$ independent real normal variables, we have $X = x^2 + y^2$.
    Recall that if $z \sim \mathcal{N}(0,\sigma^2)$, then $z^2 \sim \operatorname{Gamma}(1/2, 2 \sigma^2)$.
    Thus $x^2, y^2 \sim \operatorname{Gamma}(1/2,1)$ independently. Therefore $X = x^2 + y^2 \sim \operatorname{Gamma}(1,1)$. Similarly, for
    \begin{equation*}
        Y \coloneqq \sum_{i=2}^d |z_i|^2,
    \end{equation*}
    we have $Y \sim \operatorname{Gamma}(d-1,1)$. Consequently,
    \begin{equation*}
        |\langle u, g \rangle|^2 = \frac{X}{X+Y}.
    \end{equation*}
    Recall that if $z_1, \ldots, z_n \sim \operatorname{Gamma}(\alpha_i, \theta)$ independently, then the normalized vector $(z_1/S, \ldots, z_n/S)$ with $S = \sum_i z_i$ follows the Dirichlet distribution $\operatorname{Dir}(\alpha_1, \ldots, \alpha_n)$. Applying this with $(X,Y)$ yields
    \begin{equation*}
        \left( \frac{X}{X+Y}, \frac{Y}{X+Y} \right) \sim \operatorname{Dir}(1, d-1).
    \end{equation*}
    Since the marginal of a Dirichlet distribution is Beta-distributed, namely if $(x_1,\ldots,x_n) \sim \operatorname{Dir}(\alpha_1,\ldots,\alpha_n)$ then $x_i \sim \operatorname{Beta}(\alpha_i, \sum_j \alpha_j - \alpha_i)$, we obtain
    \begin{equation*}
        \frac{X}{X+Y} \sim \operatorname{Beta}(1,d-1).
    \end{equation*}
    The density of $\operatorname{Beta}(\alpha,\beta)$ is given by
    \begin{equation*}
        f(x) = \frac{\Gamma(\alpha+\beta)}{\Gamma(\alpha)\Gamma(\beta)} x^{\alpha-1} (1-x)^{\beta - 1},
    \end{equation*}
    Using the identity $\tfrac{\Gamma(d)}{\Gamma(1)\Gamma(d-1)} = d-1$, it follows that
    \begin{equation*}
        \Pr_U \bigg[ \frac{X}{X+Y} \geq t \bigg] = (1-t)^{d-1}.
    \end{equation*}
    Finally, we conclude that
    \begin{equation*}
        \Pr_U \Big[ |\langle v, U w \rangle|^2 \geq t \Big] = {\left( 1 - \frac{t}{\norm{v}^2 \norm{w}^2} \right)}^{d-1}.
    \end{equation*}
\end{proof}

We are now in a position to establish the result stated in \Cref{thm:tamper_dection_against_any_constants}.

\begin{proof}[{Proof of \Cref{thm:tamper_dection_against_any_constants}}]
    For any state $\ket{\psi} \in \mathrm{States}$, it holds that
    \begin{equation*}
        \Pr \big[\dec \circ \Phi_\psi \circ \enc(m) = \perp \big] = 1 - \Pr \big[ \dec \circ \Phi_\psi \circ \enc(m) = m' \text{ for some $m'$ } \big].
    \end{equation*}
    Moreover, for any Haar random scheme $(\dec_U, \enc_U)$ and any unitary $U$, we have that $\Pr [ \dec_U \circ \Phi_\psi \circ \enc_U (m) = m']$ is equals to $| \bra{m'}\bra{0} U \ket{\psi} |^2$.
    Using the union bound on \Cref{lem:unit_sphere_scalar_probabilty} we obtain
    \begin{equation*}
        \Pr_U \Big[ \exists \psi, \exists m' \ : \ | \bra{m'}\bra{0} U \ket{\psi} |^2 \geq t \Big] \leq |\mathcal{M}|\,|\mathrm{States}| {\left( 1 - t^2 \right)}^{d-1}.
    \end{equation*}
    Let $s$ be such that $t = (\log |\mathcal{M}| + \log |\mathrm{States}| + s) / (d-1)$, then
    \begin{equation*}
        \Pr_U \Big[ \exists \psi, \exists m' \ : \ | \bra{m'}\bra{0} U \ket{\psi} |^2 \geq t \Big] \leq |\mathcal{M}|\,|\mathrm{States}| \exp(-t(d-1)) = \exp(\shortminus s)
    \end{equation*}
    So with probability at least $1 - \exp(\shortminus s)$ over unitaries $U$, for all $\psi$ and all $m'$ we have $| \bra{m'}\bra{0} U \ket{\psi} |^2 \leq t$. Fix such a $U$, then for any $\psi$
    \begin{equation*}
        \Pr \big[\dec_U \circ \Phi_\psi \circ \enc_U (m) \neq \perp \big] = \sum_{m'} | \bra{m'}\bra{0} U \ket{\psi} |^2 \leq |\mathcal{M}| t.
    \end{equation*}
    that is
    \begin{equation*}
        |\mathcal{M}| t = |\mathcal{M}| \frac{\log |\mathcal{M}| + \log |\mathrm{States}| + s}{d-1} = |\mathcal{M}| \frac{k + \log |\mathrm{States}| + s}{2^n-1}.
    \end{equation*}
    Choosing $s = k$, and $n = \gamma k$ for a constant $\gamma \geq  \frac{2}{1-\alpha}$, and with $|\mathrm{States}| \leq 2^{d^{\alpha}}$, this bound is negligible in $k$.
    Thus, there exists a unitary $U$ such that for all messages $m \in \mathcal{M}$ and all $\Phi \in \mathrm{States}$,
    \begin{equation*}
        \Pr \big[\dec_U \circ \Phi_\psi \circ \enc_U (m) = \perp \big] \geq 1 - \operatorname{negl}(k).
    \end{equation*}
\end{proof}

\section{Tamper detection for quantum messages 
}

\label{sec:quantum messages}

In this section, we expand the Haar random quantum encoding-decoding scheme of classical messages (\Cref{def:Haar_random_scheme_family}), to quantum messages:
\begin{equation*}
    \mathcal{M} \coloneqq \big\{ \psi \in \CC^k \,:\, \braket{\psi}{\psi} = 1 \big\}
\end{equation*}

\begin{definition}[$\delta$-net (Lemma II.4, \cite{HLSW04})]\label{epsnets}
Let $0<\delta<1$ and let $\dim\mathcal{H}=K$.  There exists a set $\mathcal{N}_\delta$ of pure states in $\mathcal{H}$ with $
|\mathcal{N}_\delta|\le\Big(\frac{5}{\delta}\Big)^{2K},
$
such that for every pure state $|\varphi\rangle\in\mathcal{H}$ there exists $|\tilde\varphi\rangle\in\mathcal{N}_\delta$ satisfying
$
\big\|\,|\varphi\rangle\langle\varphi|-|\tilde\varphi\rangle\langle\tilde\varphi|\,\big\|_{1}\le\delta.
$
We call any such set an \(\delta\)-net.
\end{definition}

The encoding of the message remains the same.
The decoding only measures the register $B$ to detect tampering.

\begin{definition}[Haar Random Scheme Family] \label{def:Haar_random_scheme_family (quantum messages)}
    Let $\mathrm{HR}(d,k)$ denote the \emph{Haar random quantum encoding–decoding scheme family} with parameters $d$ and $k$, where $k \leq \log(d)$.
    Let $\mathcal{M}$ be a $\delta$-net.
    Let $U \in \mathcal{U}_d(\CC)$ be a $d \times d$ unitary matrix sampled from the Haar measure over the unitary group, and define $(\dec_U, \enc_U)$ to be the \emph{Haar random quantum encoding–decoding scheme} associated with $U$  as follows:
    \begin{description}
        \item[Encoding:] The register $A$ will hold the message $m$ which will be padded on register $B$ with $\log(d)-k$ sized ancilla.
        \[ \enc_U(m) \coloneqq U \big( \ket{m}^A\otimes\ket{0}^{B} \big). \]
        Fix any orthonormal basis $\{a_i\}_i$ for $A$ then, $\mathcal{C}$:= \textrm{span} $\lbrace \enc_U(\ket{a_i}) : i \rbrace$ is called the code-space. Since $U$ is a unitary operator, the code-space is invariant under the choice of basis. 
        
        \item[Decoding:] On receiving the state $\rho$, the decoder first reverts the encoder using $U^{*}$, and then measures $B$ in the computational basis. 
        Let the measurement result be denoted as $b$.
        \begin{itemize}
            \item If $b \neq 0^{\vert B \vert}$, then output $\widehat{m} = \perp$, indicating that there was tampering.
            \item If $b =  0^{\vert B \vert}$, then $\mathsf{accept}$ and output register $A$ as $\widehat{m}$.
        \end{itemize}
       The decoder can alternatively be described as follows: measure the received codeword using the \PVM  $\{ \Pi_{\mathcal{C}}, \Pi_{\perp} = \id -  \Pi_{\mathcal{C}} \}$, where $ \Pi_{\mathcal{C}}$ is the orthogonal projector onto the code-space, and $\Pi_\perp$ projects onto the orthogonal complement of the code space.
    \end{description}

Again, it is straightforward to check that the scheme is \emph{complete}, and so we move on to the security. To establish the proof, we will restate and adapt the procedure outlined in \Cref{sec:tamper_detection} for classical messages.

\begin{theorem}[Tamper detection for quantum messages]
\label{thm:tamper_detection_quantum}
    Let $\mathcal{M} \coloneqq \{ \psi \in \CC^k \,:\, \braket{\psi}{\psi} = 1 \}$, and let $\advCPTP$ be an adversarial family of \CPTP maps on $\mathcal{B}(\CC^d)$ satisfying, for parameters $0\le 2\alpha<\delta<1$, the same three conditions as in \Cref{thm:tamper_detection_code}:
    \begin{itemize}
        \item The cardinality of the adversarial family is bounded as $ \vert \advCPTP \vert \leq 2^{d^{\alpha}}$.
        \item The maximal minimal Kraus rank within the family satisfies $\rank(\advCPTP) \leq d^{1-\delta}$.
        \item The maximal entanglement fidelity of the family satisfies $F_e(\advCPTP) \leq \frac{2}{d^{\delta/2}}$.
    \end{itemize}
    Fix any net radius $\eta \in (0,1)$ and let $\mathcal{N}_\eta \subseteq \mathcal{M}$ be an $\eta$-net as in \Cref{epsnets}. Then, 
     \begin{equation*}
         \Pr_U \big[ (\dec_U, \enc_U) \text{ is $(\varepsilon+\eta/2)$-tamper secure against $\advCPTP$} \big] \geq 1 - \operatorname{negl}(k),
     \end{equation*}
     with $\varepsilon = \operatorname{negl}(k)$ and $\gamma= \mathcal{O}(1)$. In particular, there exists a unitary $U \in \mathcal{U}_d(\CC)$ such that $(\dec_U, \enc_U)$ is $(\varepsilon+\eta/2)$-tamper secure against $\advCPTP$.
\end{theorem}
\begin{proof}
    Fix any $\ket{v} \in \mathcal{M}$ and let $\ket{\varphi} = U \ket{v}^A \ket{0}^B$.
    Then, define
    \begin{equation*}
        Z_{\varphi} = \Tr \left[ \Pi_{\mathcal{C}} \Phi \left( \ketbra{\varphi}{\varphi} \right) \right].
    \end{equation*}
    Furthermore, fix any orthonormal basis $\{a_j\}_j$ for $A$ and define $\ket{\psi_{j}} = U \left( \ket{a_j} \ket{0}\right)$. Now,
    \begin{align*}
        Z_{\varphi} & =  \Tr \left[ \Pi_{\mathcal{C}} \Phi \left( \ketbra{\varphi}{\varphi} \right) \right] \\
        & =  \Tr \left[ \sum_{j} \ketbra{\psi_j}{\psi_j} \Phi \left( \ketbra{\varphi}{\varphi} \right) \right] \\
        & =  \Tr \left[ \sum_{j} \ketbra{\psi_j}{\psi_j} \sum^r_{i=1} K_i  \ketbra{\varphi}{\varphi}  K_i^*\right] \\
        & =  \sum_{i,j} \Tr \Big[  \ketbra{\psi_j}{\psi_j}  K_i  \ketbra{\varphi}{\varphi}   K_i^*\Big] \\
        & =  \sum_{i,j} \bra{\psi_j} K_i \ket{\varphi} \bra{\varphi} K_i^*\ket{\psi_j} \\
        & = \sum_{i,j} \bra{\psi_j} \bra{\varphi} K_i \otimes K_i^* \ket{\varphi} \ket{\psi_j}. 
    \end{align*}
    Raising it to the $n$-th power yields,
    \begin{align*}
        Z^n_{\varphi} & = \prod_{t=1}^n \sum_{i,j} \bra{\psi_j} \bra{\varphi} K_i \otimes K_i^* \ket{\varphi} \ket{\psi_j} \\
        & =  \sum_{i \in [r]^n} \sum_{j \in \dim(A)^n }\prod_{t=1}^n  \bra{\psi_{j_t}} \bra{\varphi} K_{i_t} \otimes K_{i_t}^* \ket{\varphi} \ket{\psi_{j_t}} \\
        & =  \sum_{i \in [r]^n} \sum_{j \in \dim(A)^n }  \Big( \tens{t=1}{n}\bra{\psi_{j_t}} \bra{\varphi} \Big) \Big(\tens{t=1}{n} (K_{i_t} \otimes K_{i_t}^*)\Big) \Big(\tens{t=1}{n} \ket{\varphi} \ket{\psi_{j_t}}\Big) \\
        & = \sum_{i \in [r]^n} \sum_{j \in \dim(A)^n} 
        \left( \bigotimes_{t=1}^{n} \bra{a_{j_t}} U^* \otimes \bra{v}\bra{0} U^* \right)
        \left( \bigotimes_{t=1}^{n} (K_{i_t} \otimes K_{i_t}^*) \right)
        \left( \bigotimes_{t=1}^{n} U \ket{v}\ket{0} \otimes U \ket{a_{j_t}} \right) \\
        & = \sum_{i \in [r]^n} \sum_{j \in \dim(A)^n} \Big[ \tens{t=1}{n}  \bra{a_{j_t}} \otimes \bra{v} \bra{0} \Big]  {U^{*}}^{\otimes 2n}  X(i) U^{\otimes 2n} \Big[ \tens{t=1}{n} \ket{v}\ket{0} \otimes \ket{a_{j_t}} \Big],
    \end{align*}
    where $X(i) = \big( K_{i_1} \otimes K_{i_1}^* \big) \otimes \cdots \otimes \big( K_{i_n} \otimes K_{i_n}^* \big)$.
    
    Let us first compute the moments $\mathbb{E}_U (Z_\varphi)^n$, similarly to the proof of \Cref{lem:nth_moment_upperbound} for classical messages. Using \Cref{lem:weingarten_calculus_bound} we have
    \begin{equation}\label{eq:moment-bound-start}
        \mathbb{E}_U\!\big[(Z_\varphi)^n\big]
        \le
        \Big(1+\mathcal{O}\big(\tfrac{4n^2}{d}\big)\Big)\,\frac{1}{d^{2n}}
        \sum_{i\in[r]^n} \sum_{j\in \dim(A)^n} \sum_{\pi\in\mathfrak S_{2n}}
        \Tr\!\big[V(\pi)^{-1} X(i)\big]\;
        S_{\pi}(\ket{v} \ket{0},j),
    \end{equation}
    with the scalars
    \begin{equation*}
        S_{\pi}(\ket{v} \ket{0},j)\coloneqq
        \Big(\bigotimes_{t=1}^n \bra{v} \bra{0}\otimes\bra{{a_{j_t}}}\Big)
        \, V(\pi) \,
        \Big(\bigotimes_{t=1}^n \ket{v} \ket{0}\otimes\ket{{a_{j_t}}}\Big).
    \end{equation*}
    Note that, $\vert S_{\pi}(\ket{v} \ket{0},j) \vert \leq 1$. So that we get
    \begin{equation*}
        \mathbb{E}_U\!\big[(Z_\varphi)^n\big]
        \le \sum_{j\in \dim(A)^n}
        \bigg[ \Big( 1 + \mathcal{O}\big(\tfrac{4n^2}{d}\big) \Big) \, \frac{1}{d^{2n}} \sum_{i\in[r]^n} \sum_{\pi\in\mathfrak S_{2n}} \Tr\!\big[V(\pi)^{-1} X(i)\big] \bigg],
    \end{equation*}
    where the bracket is similar to  \cref{eq:high_moment_s_eq_t_step} in the proof of \Cref{sec:tamper_detection}. Using the same arguments, we derive the bound
    \begin{align}
        \mathbb{E}_U\!\big[(Z_\varphi)^n\big]
        &\le \sum_{j\in \dim(A)^n} \bigg[ \Big( 1 + \mathcal{O}\big(\tfrac{4n^2}{d}\big) \Big) {\Big( \sqrt{r} \frac{\phi^2 d + 2n}{\phi d} \Big)}^{2n} \bigg] \notag \\
        &\le \Big( 1 + \mathcal{O}\big(\tfrac{4n^2}{d}\big) \Big) \dim(A)^n  {\Big( \sqrt{r} \frac{\phi^2 d + 2n}{\phi d} \Big)}^{2n}. \label{eq:moment-bound-end}
    \end{align}

    Now let us calculate the deviation bound $\Pr_U [Z_\varphi \geq \varepsilon] $, similarly to the proof of \Cref{lem:deviation_bound} for classical messages. Applying Markov’s inequality yields
    \begin{equation*}
        \Pr_U [Z_\varphi \geq \varepsilon] = \Pr_U [e^{\theta Z_\varphi} \geq e^{\theta \varepsilon}] \geq \frac{\mathbb{E}_U [e^{\theta Z_\varphi}]}{e^{\theta \varepsilon}}.
    \end{equation*}
    Therefore,
    \begin{align}
        \Pr [Z_\varphi \geq \varepsilon]
        &\leq  e^{-\theta \varepsilon} \mathbb{E}\left[ e^{\theta Z_\varphi} \right] \notag \\
        &\leq e^{-\theta \varepsilon} \sum_{i=0}^{\infty} \frac{\theta^i}{i!} \mathbb{E}[Z_\varphi^i] \notag \\
        &\leq e^{-\theta \varepsilon} \left[ \sum_{i \leq \theta_0} \frac{\theta^i}{i!} \mathbb{E}[Z_\varphi^i] + \sum_{i \geq \theta_0 + 1} \frac{\theta^i}{i!} \mathbb{E}[Z_\varphi^i] \right], \label{eq:deviation-bound-start}
    \end{align}
    where the bracket is similar to \cref{eq:deviation_bound_ss_case} in the proof of \Cref{sec:tamper_detection}. Using the same arguments, we derive the low and high-order moments bounds
    \begin{equation*}
        \sum_{i \leq \theta_0} \frac{\theta^i}{i!} \mathbb{E}[Z_{\varphi}^i] \leq \mathcal{O} \left( \exp\left( \dim(A) \frac{r \theta \left( \phi^2 d + 2\theta_0 \right)^2}{d^2 \phi^2} \right) \right),
    \end{equation*}
    and
    \begin{equation*}
        \sum_{i \geq \theta_0 + 1} \frac{\theta^i}{i!} \mathbb{E}\left[Z_{\varphi}^i\right] \leq \mathcal{O}\left( e^{\theta} \dim(A)^{\theta_0} \bigg( \sqrt{r} \frac{\phi^2 d + 2\theta_0}{\phi d} \bigg)^{2\theta_0} \right),
    \end{equation*}
    so we conclude the bound
    \begin{equation} \label{eq:deviation-bound-end}
        \Pr [Z_{\varphi} \geq \varepsilon] \leq e^{-\theta \varepsilon} \mathcal{O} \left( \exp\left( \dim(A) \frac{r \theta \left( \phi^2 d + 2\theta_0 \right)^2}{d^2 \phi^2} \right) +e^{\theta} \dim(A)^{\theta_0} \bigg( \sqrt{r} \frac{\phi^2 d + 2\theta_0}{\phi d} \bigg)^{2\theta_0}  \right).
    \end{equation}

    Finally let us prove, that for all messages $s \in \mathcal{M}$, and for all \CPTP map $\Phi \in \advCPTP$, the following holds:
    \begin{equation*}
        \Pr \left[ \dec_U \circ {\Phi} \circ \enc_U(s) = \perp \right] \geq 1 - \varepsilon,
    \end{equation*}
    where the probability is over the randomness of the measurement of $\dec_U$. For that, similarly to \Cref{sec:tamper_detection_code_proof}, we will bound the probability that the decoder fails to detect the tampering. We have
    \begin{equation*}
        \Pr_U \Big[ \Pr \big[ \dec_U \circ \Phi \circ \enc_U(s) \neq \perp \big] \geq \varepsilon \Big]
        = \Pr_U \Big[ Z_\varphi \geq \varepsilon \Big],
    \end{equation*}
    where the inner probabilities are taken over the inherent quantum random dues to quantum decoders (measurements). Using \cref{eq:deviation-bound-end}
    \begin{align*}
        &\Pr_U \Big[ \Pr \big[ \dec_U \circ \Phi \circ \enc_U(s) \neq \perp \big] \geq \varepsilon \Big] \\
        &\quad\leq e^{-\theta \varepsilon} \mathcal{O} \left( \exp\left( \dim(A) \frac{r \theta \left( \phi^2 d + 2\theta_0 \right)^2}{d^2 \phi^2} \right) +e^{\theta} \dim(A)^{\theta_0} \bigg( \sqrt{r} \frac{\phi^2 d + 2\theta_0}{\phi d} \bigg)^{2\theta_0} \right).
    \end{align*}
    Taking union bound over $s \in \mathcal{N}_\eta$ and $\Phi \in \advCPTP$, we obtain
    \begin{align*}
        &\Pr_U \Big[ \exists s, \Phi \ : \ \Pr \big[ \dec_U \circ \Phi \circ \enc_U(s) \neq \perp \big] \geq \varepsilon \Big] \\
        &\quad\leq \vert \mathcal{N}_\eta \vert \, \vert \advCPTP \vert \, e^{-\theta \varepsilon} \mathcal{O} \left( \exp\left( \dim(A) \frac{r \theta \left( \phi^2 d + 2\theta_0 \right)^2}{d^2 \phi^2} \right) +e^{\theta} \dim(A)^{\theta_0} \bigg( \sqrt{r} \frac{\phi^2 d + 2\theta_0}{\phi d} \bigg)^{2\theta_0} \right).
    \end{align*}
    Setting $r \leq d^{1-\delta}$ and $\phi = \sqrt{\frac{2 \theta_0}{d}}$, with $\delta \in [0,1)$ together with $\theta = d^{\beta}$ and $\theta_0 = d^{\delta/2}$, where $\beta$ is a constant such that $\alpha < \beta \leq \delta/2$, we obtain
    \begin{align*}
        &\Pr_U \Big[ \exists s, \Phi \ : \ \Pr \big[ \dec_U \circ \Phi \circ \enc_U(s) \neq \perp \big] \geq \varepsilon \Big] \\
        &\quad\leq \vert \mathcal{N}_\eta \vert \, \vert \advCPTP \vert \, e^{- d^\beta \varepsilon} \mathcal{O} \Bigg(  \exp\left( {8 \dim(A) d^{\beta - \delta/2 } } \right) +   e^{d^{\beta}} \bigg(\dim(A) \frac{8}{d^{\delta/2}} \bigg)^{d^{\delta/2}} \Bigg).
    \end{align*}
    Set $\beta = \frac{\delta}{2}$ and $\varepsilon = c(\eta) \dim(A) d^{-\delta/2}$, with $c(\eta) > 8 + 2 \log(5/\eta)$ then
    \begin{equation*}
        e^{- d^\beta \varepsilon} \exp\left( {8 \dim(A) d^{\beta - \delta/2}} \right) = \exp \big( {(8-c(\eta)) \dim(A)} \big),
    \end{equation*}
    and
    \begin{equation*}
        e^{d^{\beta}} {\left(\dim(A) \frac{8}{d^{\delta/2}} \right)}^{d^{\delta/2}} \leq {\left(\dim(A) \frac{22}{d^{\delta/2}} \right)}^{d^{\delta/2}}
    \end{equation*}
    Using $| \mathcal{N}_\eta | \leq (5/\eta)^{2 \dim(A)}$, $\vert \advCPTP \vert = \exp( d^\alpha )$ and $\dim(A) = 2^k$, we have
    \begin{align*}
        &\Pr_U \Big[ \exists s, \Phi \ : \ \Pr \big[ \dec_U \circ \Phi \circ \enc_U(s) \neq \perp \big] \geq \varepsilon \Big] \\
        &\quad\leq \exp\left( d^\alpha - 2^k ( c(\eta) - 8 - 2 \log \tfrac{5}{\eta} ) \right) + {\left(\dim(A) \frac{22}{d^{\delta/2}} \right)}^{d^{\delta/2}},
    \end{align*}
    which is negligible in $k$ provided $22 \dim(A) < d^{\delta/2}$.
    Thus with high probability, the bound holds for all $s \in \mathcal{N}_\eta$ and all $\Phi \in \advCPTP$. For arbitrary $s \in \mathcal{M}$, pick $\tilde{s} \in \mathcal{N}_\eta$ with $\norm[\big]{\ketbra{s}{s} - \ketbra{\tilde{s}}{\tilde{s}}}_1 \leq \eta$. By Helstrom’s bound
    \begin{equation*}
        \Tr \Big[ \Pi_{\mathcal{C}} \big( \Phi(\ketbra{s}{s}) - \phi(\ketbra{\tilde{s}}{\tilde{s}} ) \big) \Big] \leq \tfrac{1}{2} \norm[\big]{\Phi(\ketbra{s}{s}) - \phi(\ketbra{\tilde{s}}{\tilde{s}})}_1,
    \end{equation*}
    and by the contractivity of trace distance under \CPTP map,
    \begin{equation*}
        \norm[\big]{\Phi(\ketbra{s}{s}) - \phi(\ketbra{\tilde{s}}{\tilde{s}})}_1 \leq \norm[\big]{\ketbra{s}{s} - \ketbra{\tilde{s}}{\tilde{s}}}_1,
    \end{equation*}
    such that
    \begin{equation*}
        \big| Z_{\ketbra{s}{s}} - Z_{\ketbra{\tilde{s}}{\tilde{s}}} \big| \leq \frac{\eta}{2}.
    \end{equation*}
    Thus $Z_{\ketbra{s}{s}} \leq Z_{\ketbra{\tilde{s}}{\tilde{s}}} + \frac{\eta}{2} \leq \varepsilon + \frac{\eta}{2}$ for all $\Phi \in \advCPTP$. Hence the decoder rejects with probability at least $1 - (\varepsilon + \frac{\eta}{2})$ for every pure message, and by convexity for all quantum messages, completing the proof.
\end{proof}

\end{definition}

\section*{Acknowledgements}  We acknowledge the support of the Natural Sciences and Engineering Research Council of Canada (NSERC)(ALLRP-578455-2022), the Air Force Office of Scientific Research under award number FA9550-20-1-0375 and of the Canada Research Chairs Program.

\bibliographystyle{bibtex/bst/alphaarxiv.bst}
\bibliography{bibtex/bib/quasar-full.bib,
              bibtex/bib/quasar.bib,
              bibtex/bib/quasar-more-tamper-new.bib,
              bibtex/bib/quasar-more-merged.bib}

\appendix
\newpage

\section{The symmetric group} \label{app:symmetric_group}

Let $\mathfrak{S}_n$ denote the symmetric group on $n$ elements. Every element $\pi \in \mathfrak{S}_n$ admits a unique decomposition disjoint cycles, called the \emph{cycle decomposition} of $\pi$:
\begin{equation*}
    \pi = (c_1, \ldots, c_l).
\end{equation*}
The number of cycles in the cycle decomposition of a permutation $\pi$ is denoted $\#\pi$, and the set of fixed points of $\pi$, is denoted $\mathrm{Fix}(\pi)$.

\begin{lemma} \label{lem:fixed_points_bound}
    For all permutation $\pi \in \mathfrak{S}_n$ we have
    \begin{equation*}
        \mathrm{Fix}(\pi) \geq 2 \#\pi - 2n.
    \end{equation*}
\end{lemma}
\begin{proof}
    Observe that each cycle of length at least $2$ in the cycle decomposition of $\pi$ involves at least $2$ elements, while contributing exactly one to the total number of cycles $\#\pi$. Since there are $n$ elements in total, the number of such non-trivial cycles is at most $n - \mathrm{Fix}(\pi)$. It follows that
    \begin{equation*}
        \#\pi \leq \mathrm{Fix}(\pi) + \frac{n - \mathrm{Fix}(\pi)}{2}.
    \end{equation*}
    Rearranging this inequality yields the claimed lower bound: $\mathrm{Fix}(\pi) \geq 2 \#\pi - 2n$.
\end{proof}

Let $\tilde{\mathfrak{S}}_{2n} \subseteq \mathfrak{S}_{2n}$ be the subset of the symmetric group on $2n$, defined by
\begin{equation*}
    \tilde{\mathfrak{S}}_{2n} \coloneqq \set[\big]{ \pi \in \mathfrak{S}_{2n} \;:\; i \oplus \pi(i) = 1}.
\end{equation*}
That is, the permutations that map even indices to odd indices, and off indices to even indices.

\begin{lemma} \label{lem:parity_alternating_permutation_property}
    The following equality on $\tilde{\mathfrak{S}}_{2n}$ holds:
    \begin{equation*}
        \sum_{\pi \in \tilde{\mathfrak{S}}_{2n}} x^{\#\pi} = n! \cdot x^{\bar{n}},
    \end{equation*}
    where $x^{\bar{n}}$ is the raising factorial $x (x + 1) \cdots (x + n - 1)$.
\end{lemma}
\begin{proof}
    Let the Even and Odd sets
    \begin{equation*}
        E \coloneqq \set{2, 4, \ldots, 2n} \qquad \text{and} \qquad 0 \coloneqq \set{1, 3, \ldots, 2n-1}.
    \end{equation*}
    Let $\pi \in \tilde{\mathfrak{S}}_{2n}$ and define the restriction of $\pi$ to $E$ and $O$ by
    \begin{equation*}
        \alpha \coloneqq \pi|_E: E \to O \qquad \text{and} \qquad \beta \coloneqq \pi|_O: O \to E.
    \end{equation*}
    Let $\sigma \in \mathfrak{S}(E) \simeq \mathfrak{S}_n$ be defined by
    \begin{equation*}
        \sigma(i) \coloneqq \beta \big( \alpha (i) \big).
    \end{equation*}
    Then the cycle in $\pi$ containing $i$ is
    \begin{equation*}
        \big( i, \alpha(i), \sigma(i), \alpha(\sigma(i)), \sigma(i)^2, \ldots \big).
    \end{equation*}
    Thus each cycle of $\sigma$ of length $l$ generates a cycle of $\pi$ of length $2l$, \ie $\#\pi = \#\sigma$. Since there are $n!$ choice for $\alpha$, we have
    \begin{equation*}
        \sum_{\pi \in \tilde{\mathfrak{S}}_{2n}} x^{\#\pi} = n! \sum_{\sigma \in \mathfrak{S}_{n}} x^{\#\sigma} = n! \sum_{k=0}^n \genfrac{[}{]}{0pt}{0}{n}{k} \cdot x^k = n! \cdot x^{\bar{n}},
    \end{equation*}
    where $\genfrac{[}{]}{0pt}{1}{n}{k}$ are the unsigned Stirling numbers of the first kind.
\end{proof}

The symmetric group $\mathfrak{S}_n$ acts naturally on the tensor power $(\CC^d)^{\otimes n}$ via the \emph{tensor representation} $V(\cdot)$. For each $\pi \in \mathfrak{S}_n$, we define the linear operator $V(\pi)$ on $(\CC^d)^{\otimes n}$ by its action on elementary tensors:
\begin{equation*}
    V(\pi)(v_1 \otimes \cdots \otimes v_n) \coloneqq v_{\pi^{\shortminus 1}(1)} \otimes \cdots \otimes v_{\pi^{\shortminus 1}(n)},
\end{equation*}
with linear extension to the entire space. This yields a unitary representation of $\mathfrak{S}n$ on $(\CC^d)^{\otimes n}$, where the action corresponds to permuting tensor factors.

A useful identity in quantum information theory is the \emph{swap trick}:
\begin{equation} \label{eq:swap_trick}
    \Tr \left[ (A \otimes B) F \right] = \Tr [ A \cdot B ],
\end{equation}
where $F \coloneqq V((1\,2)) $ is the \emph{swap operator}, defined by $F(x \otimes y) = y \otimes x$. This identity generalizes naturally to arbitrary permutations $\pi \in \mathfrak{S}_n$ with cycle decomposition $\pi = (c_1, \ldots, c_l)$: for operators $M_1, \dots, M_n$, and the representation $V(\pi)$, one has
\begin{equation} \label{eq:generalized_swap_trick}
    \Tr \left[ (M_1 \otimes \cdots \otimes M_n) V(\pi) \right] = \prod_{i = 1}^l \Tr \left[ M_{c_i(1)} \cdots M_{c_i(l_i)} \right],
\end{equation}
where $c_i(k)$ denotes the $k$-th element of the $i$-th cycle, and $l_i$ denotes its length.

\section{Weingarten calculus} \label{app:weingarten_calculus}

\emph{Weingarten calculus} provides a powerful framework for computing averages over the unitary group $\mathcal{U}_d(\mathbb{C})$ with respect to the Haar measure. A fundamental object in quantum information theory is the \emph{moment operator}, defined as
\begin{equation*}
    \E_U \left[ U^{\otimes k} X (U^*)^{\otimes k} \right],
\end{equation*}
where $X \in \mathcal{M}_{d^k}(\mathbb{C})$, and the expectation is taken over a Haar-distributed unitary matrix $U \in \mathcal{U}_d(\mathbb{C})$. This quantity induces a \emph{twirling channel} acting on the $k$-fold tensor power representation of the unitary group.

In order to compute such averages explicitly, it is often convenient to work at the level of individual matrix elements. The central object of interest is the integral over products of unitary matrix entries and their complex conjugates:
\begin{equation*}
    \int U_{i_1 j_1} \cdots U_{i_k j_k} \bar{U}_{i_1' j_1'} \cdots \bar{U}_{i_k' j_k'} \, dU.
\end{equation*}
\emph{Weingarten calculus} yields an exact expression for this integral in terms of a double sum over permutations $\sigma, \tau \in \mathfrak{S}_k$, the symmetric group on $k$ elements:
\begin{theorem*}[{Weingarten calculus \cite{CS06}}] \label{thm:Weingarten_calculus}
    Let $d, k$ be positive integers and $i, i', j, j'$ be four $k$-tuples of positive integers from $\{1 \ldots, d\}$. Then
    \begin{equation*}
        \int U_{i_1,j_1} \cdots U_{i_k,j_k} \bar U_{i'_1,j'_1} \cdots \bar U_{i'_k,j'_k} \operatorname{d}U = \sum_{\alpha, \beta \in \mathfrak{S}_k} \delta_{i_1,i'_{\alpha(1)}} \ldots \delta_{i_k,i'_{\alpha(k)}} \delta_{j_1,j'_{\beta(1)}} \ldots \delta_{j_k,j'_{\beta(k)}} \mathrm{Wg}_d (\alpha^{-1}\beta).
    \end{equation*}
    Where $\mathrm{Wg}_d$ denotes the \emph{unitary Weingarten function}.
\end{theorem*}

Substituting this formula into the expansion of the moment operator $\E_U[U^{\otimes k} X (U^*)^{ \otimes k}]$, yields a decomposition of the twirling channel as a linear combination of permutation operators acting on $(\mathbb{C}^d)^{\otimes k}$, with coefficients determined by the Weingarten function.

Although Weingarten calculus offers rigorous method to the evaluation of Haar integrals, the resulting expressions are, in general, combinatorially involved. As a result, exact computations are typically tractable only for moments of low order. In the following, we provide explicit evaluations for the first two moments, followed by an approximation valid in the large-$d$ regime.

\begin{lemma}[First Moments of Weingarten Calculus {\cite[Cor. 13]{Mel24}}] \label{lem:Weingarten_calculus_first_moments}
    Let $M \in \mathcal{B}(\CC^d)$ and $N \in \mathcal{B}(\CC^d \otimes \CC^d)$. Then, the first and second moments of the Weingarten calculus over the unitary group are
    \begin{equation*}
        \E_U \big[ U M U^* \big] = \frac{\Tr[M]}{d} I_d \qquad \text{and} \qquad \E_U \big[ (U \otimes U) N (U^* \otimes U^*) \big] = c_I \cdot I_{d^2} + c_F \cdot F,
    \end{equation*}
    where $F$ denotes the flip (or swap) operator acting on ${(\mathbb{C}^{d})}^{\otimes 2}$. The coefficients $c_I$ and $c_F$ are explicitly given by
    \begin{equation*}
        c_I = \frac{\Tr[N] - \frac{1}{d} \Tr[N \cdot F]}{d^2 - 1} \qquad \text{and} \qquad c_F = \frac{\Tr[N \cdot F] - \frac{1}{d} \Tr[N]}{d^2 - 1}.
    \end{equation*}
\end{lemma}

\begin{lemma}[{\cite[Lem. 4.7]{HY24arxiv}, \cite[Lem. 1]{SHH24arxiv}}] \label{lem:weingarten_calculus_bound}
    Let $\Phi_k$ and $\Psi_k$ be two Hermitian maps from $\mathcal{M}_{d^k}(\mathbb{C})$ to $\mathcal{M}_{d^k}(\mathbb{C})$ defined by
    \begin{align*}
        \Phi_k (X) &\coloneqq \int_{\mathcal{U}_d(\mathbb{C})} U^{\otimes k} \;X\; {U^*}^{\otimes k} \:\operatorname{d}U \stackrel{\tiny\substack{\text{weingarten} \\ \text{calculus} \\[1ex]}}{=} \sum_{\pi,\sigma \in \mathfrak{S}_k} \operatorname{Wg}(\pi^{\shortminus 1} \sigma, d) \Tr \big[ V(\sigma)^{-1} X \big] \cdot V(\pi) \\[1em]
        \Psi_k (X) &\coloneqq \tfrac{1}{d^k} \sum_{\pi \in \mathfrak{S}_k} \Tr \big[ V(\pi)^{\shortminus 1} X \big] \cdot V(\pi),
    \end{align*}
    where $\mathfrak{S}_k$ denotes the symmetric group on $k$ elements, and $V(\pi)$ is defined as the tensor permutation of $(\CC^d)^{\otimes k}$ associated with $\pi \in \mathfrak{S}_k$. If $d > \sqrt{6}k^{\tfrac{7}{4}}$, then we have the inequalities
    \begin{equation*}
        \Big( 1 - \mathcal{O} \big( \tfrac{k^2}{d} \big) \Big) \Psi_k \preceq \Phi_k \preceq \Big( 1 + \mathcal{O} \big( \tfrac{k^2}{d} \big) \Big) \Psi_k.
    \end{equation*}
\end{lemma}
\fi
\newpage

\end{document}